\documentclass[a4paper,UKenglish,cleveref, autoref, thm-restate, %
]{lipics-v2021}

\usepackage{xspace}

\title{Determinisation and Unambiguisation of Polynomially-Ambiguous Rational Weighted Automata} %

\titlerunning{Determinisation of Polynomially-Ambiguous Rational Weighted Automata} %

\author{Ismaël Jecker}{University of Franche-Comté, Besançon, France}{}{}{}%

\author{Filip Mazowiecki}{University of Warsaw, Poland}{}{https://orcid.org/0000-0002-4535-6508}{}%

\author{David Purser}{University of Liverpool, UK}{}{https://orcid.org/0000-0003-0394-1634}{}%

\authorrunning{I. Jecker, F. Mazowiecki and D. Purser} %

\Copyright{Ismaël Jecker, Filip Mazowiecki and David Purser} %

\usepackage{tikz}
\usetikzlibrary{snakes, positioning}



\category{} %

\relatedversion{} %

\supplement{}%

\acknowledgements{}%

\nolinenumbers %

\hideLIPIcs  %

\EventEditors{John Q. Open and Joan R. Access}
\EventNoEds{2}
\EventLongTitle{42nd Conference on Very Important Topics (CVIT 2016)}
\EventShortTitle{CVIT 2016}
\EventAcronym{CVIT}
\EventYear{2016}
\EventDate{December 24--27, 2016}
\EventLocation{Little Whinging, United Kingdom}
\EventLogo{}
\SeriesVolume{42}
\ArticleNo{23}

\newcommand*{\eg}{e.g.\@\xspace}
\newcommand*{\ie}{i.e.\@\xspace}

\newcommand{\id}{\mathsf{I\kern-0.8pt d}}

\newcommand{\A}{\mathcal{A}}
\newcommand{\B}{\mathcal{B}}
\newcommand{\C}{\mathcal{C}}
\newcommand{\D}{\mathcal{D}}
\newcommand{\U}{\mathcal{U}}
\newcommand{\T}{\mathcal{T}}
\newcommand{\runs}{\#\mathit{runs}}

\usepackage{thmtools}
\usepackage{thm-restate}
\usepackage[inline]{enumitem}
\usepackage[bibliography=common]{apxproof}

\newcommand{\ramsey}{R_\A}

\newcommand{\Z}{\mathbb{Z}}
\newcommand{\N}{\mathbb{N}}
\newcommand{\Q}{\mathbb{Q}}

\newcommand\underrel[2]{\mathrel{\mathop{#2}\limits_{#1}}}

\begin{document}

\maketitle
\begin{abstract}
We study the determinisation and unambiguisation problems of weighted automata over the rational field: Given a weighted automaton, can we determine whether there exists an equivalent deterministic, respectively unambiguous, weighted automaton?
Recent results by Bell and Smertnig show that the problem is decidable, however they do not provide any complexity bounds. We show that both problems are in PSPACE for polynomially-ambiguous weighted automata.
\end{abstract}

\section{Introduction}
Weighted automata are a popular model of computation assigning values to words~\cite{droste2009handbook}.
The domain of these values has a semiring structure
(a structure with addition and product that can define matrix multiplication).
Two popular domains are:
the field of rationals;
and the min-plus semiring over naturals extended with $+\infty$ (also known as the tropical semiring).
Depending on the domain the values have different interpretation.
The most intuitive setting, called \emph{probabilistic automata}~\cite{paz71},
is when the words are assigned their probability of being accepted
(this is a special case of the rational field domain).

We investigate the \emph{determinisation problem} (also known as the sequentiality problem):
given a weighted automaton decide if there is an equivalent deterministic weighted automaton.
If it is the case we say that the automaton is determinisable.
The problem can be stated for any semiring,
but it is very different depending on the choice of the domain (see the survey~\cite{LombardyS06}).
For example, if the semiring is the Boolean semiring then weighted automata coincide with finite automata and the problem trivialises as every automaton can be determinised.
Recently, the problem has been shown to be decidable for the rational field, but no complexity bounds were provided~\cite{abs-2209-02260}. For the tropical semiring determinisation remains an intriguing open problem:
decidability is known only for weighted automata with bounded ambiguity~\cite{KlimannLMP04,LombardyS06}.

Classes with bounded ambiguity can be defined for any automata model (not only weighted automata).
The simplest and most studied  class is the class of \emph{unambiguous automata}, where every word has at most one accepting run (but the automaton does not have to be deterministic). This class has received a lot of attention for many automata models (see \eg the survey~\cite{Colcombet15}). Unambiguous weighted automata are mathematically an elegant class of functions, as they capture functions that use only the semiring product.
A lot of research on the determinisation problem focused on the subproblem when the input automaton is unambiguous~\cite{Mohri97,KirstenM05}.
In this case automata that can be determinised are characterised by forbidding a simple pattern in the automaton, called \emph{twins property}, which can be detected in polynomial time. Due to these result papers often focus on the \emph{unambiguisation problem}, \ie whether there exist and equivalent unambiguous automaton.
The mentioned work deals with the tropical semiring and its variants.
A similar characterisation also works over the rational fields, which we describe in~\cref{appendix:twinproperty}. Thus for unambiguous weighted automata over rationals the determinisation problem is decidable in polynomial time, which we consider a folklore result.

\begin{toappendix}
\label{appendix:twinproperty}
In this section we describe the twin property for unambiguous weighted automata over the rational field. This is unnecessary for our proofs, but it gives an intuition why the unambiguisation and determinisation problems are similar.

Fix an unambiguous weighted automaton over rationals $\A= (Q, \Sigma, M, I, F)$. We say that it is trimmed if all states are reachable and coreachable (see \cref{sec:prelim}).
Note that for unambiguous weighted automata the existence of a run on word $w$ is equivalent to $\A(w) \neq 0$. This is because for unambiguous weighted automata the sum operator is not used, and thus \cref{remark:sumzero} does not apply here.
We write $q \xrightarrow{u} p$ if $M(u)_{q,p} \neq 0$. We write  $I \xrightarrow{u} p$ if there exists $r \in Q$ such that $I(r) \neq 0$ and $M(u)_{r,p} \neq 0$. In words, there is a nonzero run from an initial state to $p$ over the word $u$. Similarly, we write $p \xrightarrow{u} F$ if there exists $r \in Q$ such that $F(r) \neq 0$ and $M(u)_{p,r} \neq 0$. Note that if $\A$ is unambiguous then $r$ is unique both in runs from $I$ and runs to $F$.

We say that $\A$ has the twin property if for every $u,v \in \Sigma^*$ and $p,q \in Q$ such that $I \xrightarrow{u} p \xrightarrow{v} p$ and $I \xrightarrow{u} q \xrightarrow{v} q$ we have $|M(v)_{p,p}| = |M(v)_{q,q}|$.\footnote{Compared to standard twin property, \eg used in~\cite{Mohri97,KirstenM05} for the tropical semiring, we needed to change the weights of cycles to their absolute values.} A standard argument shows that the twin property can be detected in polynomial time by the reduction to the equality test for weighted automata (which is in polynomial time~\cite{Tzeng96,kiefer13}.)

\begin{lemma}
The twin property of an unambiguous automaton $\A$ can be decided in polynomial time.
\end{lemma}
\begin{proof}
We defined two automata $\A_1,\A_2$ operating on pairs of runs on $\A$. Intuitively, both $\A_1,\A_2$ will track and verify  the pair of runs in which $I \xrightarrow{u} p \xrightarrow{v} p$ and $I \xrightarrow{u} q \xrightarrow{v} q$. $\A_1$ will output the value of $M(v)_{p,p}$ and $\A_2$ will output the value of $M(v)_{q,q}$, which we require to be equal. We give further details:

Let $T \subseteq Q\times \Sigma\times Q\times \mathbb{Z}$ be transitions of $\A$, and so $w \in (T\otimes T)^*$ describes a pairs of runs in $\A$. We call a pair of runs $w$ valid if at each step both transitions are non-zero and take the same symbol in $\Sigma$, and the starting state of the transition complies with the previous.
We define $\A_1,\A_2$ to have non-zero value on words of the form $w_1\$w_2$, where $w_1$ and $w_2$ are both valid pairs of runs on words $u$ and $v$. We define $\A_1( w_1 \$ w_2 ) = |M_{p,p}(v)|$ where the first run in $w_1$ ends in state $p$ and $\A_2$ takes values $|M_{q,q}(v)|$ where second run on $u$ end in $q$. 

It follows that, $\A$ satisfies the twin property if and only if $\A_1 = \A_2$. Recall, deciding if $\A_1 = \A_2$ can be decided in polynomial time~\cite{Tzeng96,kiefer13}.

It remains to verify that $\A_1$ and $\A_2$ can be constructed in polynomial time. We describe $\A_1$ (as $\A_2$ is symmetric). Take $|\A| +1$ copies of $\A$: the first copy of $\A$ deterministically tracks the first run of $w_1$ and every transition has weight $1$. From state $p$, the $\$$ takes the automaton to a copy of $\A$ in which the initial and accepting state is $p$ and the weight of every transition has the absolute value of its usual weight (note that, $p$ was final state of the first run in $w_1$). Note that it is safe to take the absolute value of the weight since $\A$ is unambiguous.
\end{proof}

\begin{proposition}\label{proposition:twin}
A trimmed unambiguous weighted automaton $\A$ over the rational field is determinisable if and only if it has the twin property.
\end{proposition}

\begin{proof}
($\implies$) For a contradiction, suppose that the twin property does not hold and let $u,v \in \Sigma^*$ and $p,q \in Q$ be such that $I \xrightarrow{u} p \xrightarrow{v} p$, $I \xrightarrow{u} q \xrightarrow{v} q$ and $|M(v)_{p,p}| \neq |M(v)_{q,q}|$. We denote $x_p$ and $x_q$ to be the absolute values of unique runs in $I \xrightarrow{u} p \xrightarrow{v} p$, $I \xrightarrow{u} q \xrightarrow{v} q$, and we write $y_p = |M(v)_{p,p}|$, $y_q = |M(v)_{q,q}|$. Since $\A$ is trim there exist words $w_p$ and $w_q$ such that $p \xrightarrow{w_p} F$ and $q \xrightarrow{w_q} F$. Let $z_p, z_q$ be the absolute values of the unique runs. Then $\A(uv^nw_p) = x_pz_p\cdot y_p^n$ and $\A(uv^nw_q) = x_qz_q\cdot y_q^n$. Recall that $|y_p| \neq |y_q|$ and thus $\lim_{n \to +\infty}\frac{|\A(uv^nw_p)|}{|\A(uv^nw_q)|}$ is either $0$ or $+\infty$.

Suppose there is an equivalent deterministic automaton $\A'$, \ie $\A' = \A$. %
For every $n$ there is a unique run witnessing $I \xrightarrow{uv^n}$, let $a_n$ be its absolute value.
The words $w_p$, $w_q$ and the size of $\A'$ are fixed. Thus the values $|\A(uv^nw_p)|$ and $|\A(uv^nw_q)|$ are both equal to $a_n$ multiplied by some constants. This is a contradiction.

($\impliedby$) Suppose the twin property holds. The construction of the deterministic automaton is very similar to the second part of the proof of \cref{prop:pumpIsUnamb}. Namely, the deterministic automaton remembers unprocessed parts of the input word. Once the word becomes big enough the automaton depumps a value. We only sketch the construction as it is standard.

The size of the word stored in the automaton is bounded by $R$ such that for every word $w$ of length at least $R$ there exists an infix $v$ such that $M(v)$ has idempotent structure. Such an $R$ exists by standard Ramsey arguments, see \eg~\cite{Jecker21}. For trim unambiguous weighted automata this means that $M(v)$ is diagonal. Let $w = uvv'$. The twin property guarantees that all nonzero values on the diagonal reachable by $u$ must have the same absolute value. Thus every time we find such an idempotent $v$ we can deterministically depump this absolute value. Note that to correctly evaluate the automaton we also need extra states that allow us to keep the sign of the runs. This information can be stored in finite space, as unambiguous automata have a bounded number of runs~\cite{WeberS91} and we only need to remember the parities of negative signs.
\end{proof}

\end{toappendix}

Other popular classes of bounded ambiguity are: \emph{finitely-ambiguous automata} and \emph{polynomially-ambiguous} automata, where the number of accepting runs is bounded by a constant and a polynomial in the size of input word, respectively. Both classes are characterised by forbidding simple patterns that can be detected in polynomial time~\cite{WeberS91}. The decidability results of determinisation over the tropical semiring are for precisely these two classes~\cite{KlimannLMP04,LombardyS06}. In general bounded ambiguous classes of automata are well-known restrictions studied also in other areas~\cite{Raskin18,BarloyFLM22,CzerwinskiH22}.

We focus on the rational field, which recently gained more attention. For simplicity of explanation we present it now assuming weighted automata are over integers. In 2021 Bell and Smertnig~\cite{bell2021noncommutative} proved Reutenauer's conjecture~\cite{Reutenauer79}. It states that a weighted automaton is unambiguisable if and only the image of the automaton is a set of integers with finitely many prime divisors.\footnote{The conjecture (now theorem) is stated more generally for any field.} One implication is straightforward, as unambiguous weighted automata can only output values obtained as a product of its weights. The other implication is the core of the paper~\cite{bell2021noncommutative}. An important step is to compute the \emph{linear-hull} of the input automaton, which boils down to computing the linear Zariski closure of the semigroup generated by the input matrices. It turns out that the input automaton is unambiguisable if and only if its linear hull is unambiguous, reducing the unambiguisation problem to computing the linear hull.
The fact that it is computable essentially follows from computability of the Zariski closure (not linear)~\cite{HrushovskiOP018}. In a recent paper~\cite{abs-2209-02260} Bell and Smertnig focus on computing directly the linear Zariski closure, however, still no complexity upper bound is known.

\paragraph*{Our contribution}
Following the work of Bell and Smertnig we study the determinisation and unambiguisation problems over the rational field. Our result is the following
\begin{restatable}{theorem}{maintheorem}\label{theorem:main}
    Unambiguisability and determinisability are decidable
    in polynomial space
    for polynomially-ambiguous weighted automata.
\end{restatable}

A detailed overview of our approach is summarised in \cref{section:overview}, after introducing the necessary technical preliminaries. Roughly speaking, our approach is to reduce the problems to patterns that behave like in the unary alphabet, \ie by pumping the same infix.  Over the unary alphabet the problem boils down to analysing simple properties of linear recursive sequences (see~\cite{BarloyFLM22} or~\cite{Kostolanyi22}). For polynomially-ambiguous weighted automata we crisply characterise the automata which can be determinised/unambiguised through a notion we call pumpability. Crucially, our notion of pumpability can be decided using a standard zeroness test for weighted automata --- the automaton we require to do this is of exponential size, but combining this with an NC$^2$ algorithm for zeroness we obtain polynomial space.
Our proofs are surprising for two reasons. Both of our decision procedures for unambiguisation and determinisation are in PSPACE, but our algorithm is not effective, \ie in the case there is an equivalent deterministic/unambiguous automaton our algorithm does not construct it. An equivalent deterministic automaton can be constructed using our proof techniques, but the deterministic automaton would be of nonelementary size. We leave open whether the equivalent deterministic automaton need be this large, or whether the equivalent unambiguous automaton can be constructed using our proof techniques.

\paragraph*{Related work}
The seminal result for weighted automata over the rational field is due to Sch{\"{u}}tzenberger~\cite{Schutzenberger61b}, who proved that equivalence is decidable. The problem is known to be even in NC$^2$~\cite{Tzeng96,kiefer13} (thus in particular in polynomial time).
Other problems like containment or emptiness are undecidable already for probabilistic automata~\cite{paz71}.
Recently these undecidable problems gained attention for weighted automata with restricted ambiguity.
The goal is to determine the border of decidability between the classes: finitely ambiguous, polynomially-ambiguous, full class of weighted automata~\cite{FijalkowR017,DaviaudJLMP021,Bell22,ChistikovKMP22,CzerwinskiLMPW22}. The determinisation problem is known to be a particular variant of register minimisation, see \eg~\cite{Daviaud20}. Recently, following the work of Bell and Smertnig~\cite{bell2021noncommutative}, it was proved that register minimisation is also decidable over the rational field~\cite{abs-2307-13505}.

\section{Preliminaries}
\label{sec:prelim}

\subsection{Weighted automata}
A weighted automaton $\A$ is a tuple $(Q, \Sigma, M, I, F)$, where $Q$ is a finite set of states, $\Sigma$ is a finite alphabet, $M:\Sigma \to \mathbb{Q}^{Q\times Q}$ is a transition weighting function, and $I,F\subseteq\mathbb{Q}^Q$ are the initial and the final vectors, respectively. 

For every word $w = a_1\dots a_n$ we define the matrix $M(w)=M(a_1)M(a_2)  \ldots  M(a_n)$. We denote the empty word by $\epsilon$, and $M(\epsilon)$ is the identity matrix. For every word $w\in\Sigma^*$, the automaton outputs $\A(w) = I^T M(w) F$.	We say two automata $\A$ and $\B$ over $\Sigma$ are equivalent if $\A(w) = \B(w)$ for every $w\in\Sigma^*$

We can also interpret a weighted automaton as a finite automaton with weighted edges: when $x\ne 0$ we denote $M(a)_{p,q} = x$ by a transition $p \xrightarrow{a : x} q$. We say a state $q$ is initial if $I(q) \ne 0$ and final if $F(q) \ne 0$.  

The size of the automaton $\A$, denoted $|\A|$, is the number of states of the automaton. The norm of $\A$, denoted $||\A||$, is the largest absolute value of numerators and denominators of numbers occurring in $M,I$ and $F$. Observe the automaton can be represented using $O(|\A|^2\log(||\A||))$ bits.

\begin{definition}
    We say a matrix $M$ is \emph{p-triangular} if there exists a permutation matrix $P$  such that $P M P^{-1}$ is upper triangular.
\end{definition}

\begin{remark}
The diagonal entries of an upper-triangular matrix are exactly the eigenvalues of the matrix. Since permutations do not change which entries are on the diagonal (only their order), when we refer to diagonal entries of p-triangular matrix, we equivalently refer to the eigenvalues of the matrices.
\end{remark}

\subsection{Paths, runs, counting runs, and the monoid of structures}\label{subsection:paths}

A path of $w= a_1\dots a_n$ in $\A$ is a sequence of states $q_1,\dots,q_{n+1}$ such that for $1\le i \le n$ we have $q_i \xrightarrow{a_i : x_i} q_{i+1}$ (recall, this notation entails that $x_i \ne 0$). The value of a path is the product of the $x_i$.  A path is a cycle if $q_1 = q_{n+1}$. A path is a run if $q_1$ is initial and $q_{n+1}$ is final, and the value of the run is the product of $I(q_1),F(q_{n+1})$ and the value of the path.

A state $q$ is reachable if there exists an initial state $p$ and a word $w$ such that there is a path from $p$ to $q$. A state $q$ is coreachable if there exists a finial state $r$ and a word $w$ such that there is a path from $q$ to $r$. Henceforth, we assume every state is reachable and coreachable (by trimming the automaton where necessary).

We denote by $\runs_\A(w)$ the number of distinct runs on word $w$ in $\A$.

\begin{remark}\label{remark:sumzero}
The existence of a run from $q$ to $p$ on $w$ does not necessarily entail that $ \A(w) \ne 0$. The value of all the runs on $w$ may sum to zero, as $\A$ is not necessarily non-negative.
\end{remark}

We can also define the number of runs by treating the underlying finite automaton as a weighted automaton with weight $1$ on every non-zero transition.
Given $x\in \mathbb{Q}$, let $\overline{x} = 1$ if $x \ne0 $ and $\overline{x} = 0$ if $x = 0$.
We extend the notion point wise to the transitions, initial and final states:
$\overline{M}(a)_{p,q} = \overline{{M}(a)_{p,q}}$,
$\overline{I}_q = \overline{I_q}$ and $\overline{F}_q = \overline{F_q}$.
Then we have that $
\runs_\A(w) = \overline{I}^T \overline{M}(a_1)\overline{M}(a_2)  \ldots  \overline{M}(a_n) \overline{F}$.

For a given matrix $M$, we say that $\overline{M}$ is the \emph{structure} of the matrix.
The set of structures equipped with the matrix multiplication $\otimes$
taken in the Boolean semiring (i.e. $1+1 = 1$) forms a monoid,
often called the monoid of \emph{Boolean matrices} in the literature.
We say that a matrix $M$ has \emph{idempotent structure} if its structure is an idempotent element of this monoid:
$\overline{M} = \overline{M}\otimes\overline{M}$.

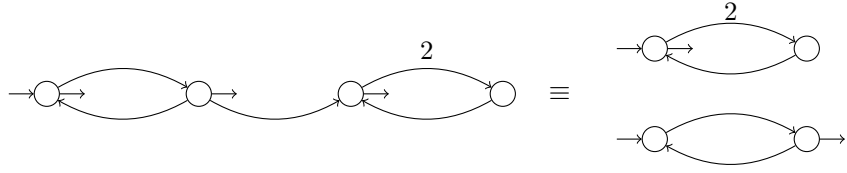
\begin{figure}
\centering
\begin{tikzpicture}
\node[circle,draw] (q1) at (-3,0) {};
\node[circle,draw] (q2) at (-1,0) {};
\node[circle,draw] at (3,0) (p1) {};
\node[circle,draw] at (1,0) (p2) {};

\node at (3.75,0) {\large $\equiv$};

\node[circle,draw]  at (5,0.6) (c1) {};
\node[circle,draw]  at (7,0.6) (c2) {};
\node[circle,draw]  at (5,-0.6) (d1) {};
\node[circle,draw]  at (7,-0.6) (d2) {};

\path
(q2) edge[->,bend right] node[below] {} (p2)
(q1) edge[->] ++(0.5,0)
(q2) edge[->] ++(0.5,0)
(q1) edge[<-] ++(-0.5,0)
(p2) edge[->] ++(0.5,0)
(p1) edge[->,bend left] node[below] {} (p2)
(p2) edge[->,bend left] node[above] {$2$} (p1)
(q1) edge[->,bend left] node[above] {} (q2)
(q2) edge[->,bend left] node[below] {} (q1)
;

\path 
(c1) edge[->] ++(0.5,0)
(c1) edge[<-] node[above]{} ++(-0.5,0)
(d1) edge[<-] ++(-0.5,0)
(d2) edge[->] ++(0.5,0)
(c1) edge[->,bend left] node[above,yshift=-0.08cm] {$2$} (c2)
(c2) edge[->,bend left] node[below,yshift=0.08cm] {} (c1)
(d1) edge[->,bend left] node[above] {} (d2)
(d2) edge[->,bend left] node[above] {} (d1)
;
\end{tikzpicture}
\caption{Example of a unary weighted automaton, where the input label of all edge is the letter~$a$ (omitted on the picture). Unlabelled edges are assumed to have weight $1$.
Observe, if $n$ is even then $\A(a^n)=1+\sum_{i=0}^{n/2 -1} 2^i = 2^{n/2}$ and if $n$ is odd then $\A(a^n)=1$.
The automaton on the left is polynomially-ambiguous,
and unambiguisable as depicted on the right, but the function is not determinisable.
}\label{fig:example}
\end{figure}

\subsection{Determinisim, ambiguity and decision problems}\label{subsection:ambiguity}
\begin{definition}
We say that an automaton $\A$ is:
\begin{itemize}
    \item \emph{deterministic} if it  has at most one non-zero entry in $I$, and $M(a)$ has at most one non-zero entry on every row for every $a\in\Sigma$,
	\item \emph{unambiguous} if $\runs_\A(w) \le 1$ for every word $w\in \Sigma^*$,
	\item \emph{finitely ambiguous} if there exists $k$ such that $\runs_\A(w) \le k$ for every word $w\in \Sigma^*$, 
	\item \emph{polynomially-ambiguous} if there exists a polynomial $\mathfrak{p}$ such that $\runs_\A(w) \le \mathfrak{p}(|w|)$ for every word $w\in \Sigma^*$, and
\item  \emph{exponentially ambiguous} otherwise, in particular, $\runs_\A(w) \le |Q|^{|w|+1}$ for every $w$.
\end{itemize}
\end{definition}
These characterisations lead to the following natural problems:
\begin{itemize}
    \item The \emph{determinisation problem} asks, given a weighted automaton $\A$, if there is an equivalent deterministic weighted automaton. If the answer is positive, we say $\A$ is \emph{determinisable}.
    \item The \emph{unambiguisation problem} asks, given a weighted automaton $\A$, if there is an equivalent unambiguous weighted automaton. If the answer is positive, we say $\A$ is \emph{unambiguisable}.
\end{itemize}
An example polynomially-ambiguous weighted automaton that is unambiguisable but not determinisable is depicted in~\cref{fig:example}.

\subsection{Closure properties}
We recall a standard result: weighted automata are closed under negation, difference and product. 

\begin{restatable}{lemma}{lemmaCombination}\label{lemma:combination}
    Let $\A_1$ and $\A_2$ be weighted automata over $\Sigma$
    that have size $m_1$ and $m_2$.
    \begin{itemize}[nolistsep]
    \item
    The function $- \A:  u \mapsto  - \A(u)$
    is recognised by an automaton of size $|\A|$ and norm $||\A||$;
    \item
    The function $\A_1 - \A_2:  u \mapsto \A_1(u) - \A_2(u)$
    is recognised by an automaton of size $|\A_1| + |\A_2|$ and norm $\max\{||\A_1||, ||\A_2||\}$;
    \item
    The function $\A_1 \cdot \A_2:  u \mapsto \A_1(u) \cdot \A_2(u)$
    is recognised by an automaton  of size $|\A_1| \cdot |\A_2|$ and norm $||\A_1|| \cdot ||\A_2||$.
    \end{itemize}
Furthermore,
if we assume that we can compute in space
$O(\log(||\A_1||)$ (respectively $O(\log(||\A_2||)$)
the weight of a transition in $\A_1$ (respectively $\A_2$)
given by a letter $a \in \Sigma$ and two states of $\A_1$ (respectively $\A_2$),
then we can compute in space $O(\log(||\A_1||\cdot||\A_2||)$
the weight of a transition in $\A_1 - \A_2$ or $\A_1\cdot\A_2$
given by a letter $a\in\Sigma$ and two states of the corresponding automaton.
\end{restatable}
\begin{toappendix}
\lemmaCombination*
\begin{proof}
Negation is obtained by taking $-I$ as the initial vector, disjoint union of two automata encodes summation, which together give the difference.

The following construction gives the product: given $(Q,\Sigma, M, I, F)$, construct $(Q\times Q, \Sigma, M',I',F')$ with $M'$ given by transitions $(q,q') \xrightarrow{a: x\cdot y} (p,p')$ whenever $q\xrightarrow{a:x} p$ in $\A_1$ and $q'\xrightarrow{a:y}p'$ in $\A_2$. For all, $q,q'\in Q$  let $I'(q,q') = I(q)I(q')$ and  $F'(q,q') = F(q)F(q')$.

Given $q,q'\in \A_1-\A_2$ the transition probability of $M_{\A_1-\A_2}(a)_{q,q'}$ can easily be computed by looking up $M_{\A_1}(a)_{q,q'}$ or $M_{\A_2}(a)_{q,q'}$ assuming $q,q'$ were from the same automaton, or returning $0$ if not. 

Given $(q,q'),(p,p')\in \A_1\cdot \A_2$ the transition probability of $M_{\A_1\cdot\A_2}(a)_{q,q'}$ can easily be computed by taking the product of  $M_{\A_1}(a)_{q,q'}$ and $M_{\A_2}(a)_{p,p'}$.
\end{proof}
\end{toappendix}

\subsection{Assumptions}\label{sec:assumptions}
In this paper, without loss of generality, we only consider weighted automata that satisfy the following two assumptions:

\subparagraph*{Non-negative transitions}
We assume that the weight of every path is non-negative,
although the run, when combined with $I,F$, may be negative.
Formally, we assume that $M(a)_{p,q} \ge 0$ for all $a\in\Sigma,p,q\in Q$.
Should the condition fail,
there is a polynomial time algorithm to produce an equivalent weighted automaton
whose matrix entries are non-negative.
This can be done so that either the initial vector or the final vector is non-negative,
so only one of these may have negative values.
The full construction is available in \cref{appendix:construction}.
\begin{toappendix}
\subsection{Constructing an equivalent automaton with non-negative transitions}
    \label{appendix:construction}
Given $\A = (Q,\Sigma, M,I,F)$, we construct equivalent $\A' = (Q',\Sigma,M',I',F')$. Let $Q'= \{q_- \mid q\in Q \} \cup\{q_+ \mid q\in Q \}$.
Given a transition $q\xrightarrow{a:x}p$ in $\A$, we add the following transitions to $\A'$:
\begin{itemize}
\item $q_+ \xrightarrow{a:x} q_+$ and $q_- \xrightarrow{a:x} q_-$ if $x\ge 0$,
\item $q_+ \xrightarrow{a:-x} q_-$ and $q_- \xrightarrow{a:-x} q_+$ if $x< 0$.
\end{itemize}
For every $q\in Q$, we let:
\begin{itemize}
    \item $I(q_+) = I(q)$ if $I(q) \ge 0$ and $0$ otherwise,
    \item $I(q_-) = -I(q)$ if $I(q) < 0$ and $0$ otherwise,
    \item $F(q_+) = F(q)$,
    \item $F(q_-) = -F(q)$.
\end{itemize} 
\begin{claim}
$\A'(w) = \A(w)$ for all $w\in \Sigma*$.
\end{claim}
\begin{proof}

Consider a run on $w$ from $q$ to $p$ in $\A$. We consider the equivalent run in $\A'$, which goes from either $q_{+}$ or $q_{-}$ to $p_+$ or $p_-$.  It is clear that the absolute value is maintained by the translation, we verify the correct sign is preserved. The sign of $I',M'$ are positive, thus the sign of the run in $\A'$ depends only on $F'$.

First suppose the number of negative transitions taken in the run is even. In this case the sign should be that of $I(q)F(q)$. If $I(q)\ge 0$, then the path goes from $q_+$ to $p_+$, which has the same sign as $F(p)$.
If $I(q)< 0$ then the path goes from $q_-$ to $p_-$, in which case the sign of both $I$ and $F$ are swapped, thus the sign of $I(q)F(q)$ is the same as $I(q_-)F(p_-)$.

Secondly suppose the number of negative transitions taken in the run is odd, in which case the sign should be swapped from that of $I(q)F(p)$. If $I(q)\ge0$, the sign should be opposite to $F(p)$, indeed the path goes from $q_+$ to $p_-$, for which the sign $F'(p_-)$ is swapped from that of $F(p)$.
If $I(q)<0$, the sign should be the same as $F(p)$, indeed the path goes from $q_-$ to $p_+$, for which the sign of $F'(p_+)$ is the same as $F(p)$.
\end{proof}

\begin{remark}
    As similar construction can assume that negative values appear only in the initial vector, rather than the final vector.
\end{remark}

\end{toappendix}

\begin{remark}\label{remark:monoid}
    A crucial consequence of this assumption is that for every word $u \in \Sigma^*$,
    each entry $(M(u))_{ij}$ of the matrix $M(u)$ is non-zero if and only if there
    exists (at least) one path labelled by $u$ between $i$ and $j$.
    Note that this is \emph{not} necessarily the case if negative weights are allowed,
    since then distinct runs can cancel each other.
    Stated more formally, this assumption implies that the function mapping a word $u \in \Sigma^*$
    to to the structure $\overline{M}(u)$
    of the corresponding matrix is a monoid homomorphism between the free monoid $\Sigma^*$ and the finite monoid
    of Boolean matrices: $\overline{M}(uv) = \overline{M}(u) \cdot \overline{M}(v)$
    for every $u,v \in \Sigma^*$.
\end{remark}

\subparagraph*{Integer values}
We assume that the weights of the automaton are integer values,
rather than rational values.
That is, $M: \Sigma \to \Z^{Q\times Q}$, $I,F\in \Z^{Q}$.

In case the condition does not hold,
an integer weighted automaton $\A'$ can be constructed such that $\A$ is unambiguisable (resp. determinisable)
if and only if $\A'$ is unambiguisable  (resp. determinisable).
Let $X$ be the set of denominators of transition weights, initial weights and final weights
and let $x = \operatorname{lcm}(X)$. 
Note that $|x| \le ||\A||^{|X|}$, thus $\log(x) \le |X|\log(||\A||)$,
and observe that $|X|,\log(||\A||)$ are both polynomial in the size of the representation of $\mathcal{A}$.

We construct $\A'$ from $\A$ with $M,I,F$ replaced by $M',I',F'$, where $I'(q) = x I(q)$,$F'(q) = x F(q)$,
and $M'(a)_{p,q} = x M(a)_{p,q}$ for all $a\in \Sigma, p,q\in Q$. Now $\A'(w) = x^{|w|+2}\A(w)$.

\section{Overview of our main result}
\label{section:overview}
We start with an overview that presents the ideas
behind our decision procedure,
the technical details are presented in the following sections.
Our algorithms deciding whether a weighted automaton
$\A = (Q, \Sigma, M, I, F)$
is unambiguisable, respectively determinisable,
rely on the study of the behaviours of $\A$
over families of words of the form
$(uv^nw)_{n \in \mathbb{N}}$,
with $u,v,w \in \Sigma^*$.
Since $\A(uv^nw)$ is defined as
$I \cdot M(u) \cdot M(v)^n \cdot M(w) \cdot F$,
understanding these behaviours reduces to understanding
powers of matrices.

\begin{restatable}{lemma}{keyLemma}\label{lemma:pumpMatrixS}
    Let $M$ be an $m \times m$ p-triangular matrix
    with the set of diagonal entries
    $\{d_1,d_2,\ldots,d_k\} \subseteq \mathbb{N}$,
    and let $\vec{x},\vec{y} \in \mathbb{Q}^m$.
    There exist $k$ polynomials $p_1,p_2,\ldots,p_k$
    such that
    \begin{equation}\label{equation:sum}
    \vec{x}^T \cdot M^n \cdot \vec{y}=
    \sum_{i=1}^k d_i^n \cdot p_i(n)
    \textup{ for all }
    n \geq m.
    \end{equation}
    Moreover, 
    if the matrix $M$ is invertible;
    $p_1,p_2,\ldots,p_k$ are all constant polynomials;
    and at most one $p_j$ is not constantly $0$,
    then
    \begin{equation}\label{equation:invertible}
    \vec{x}^T \cdot M^n \cdot \vec{y}=
    d_j^n \cdot \vec{x}^T \cdot \vec{y}
     \textup{ for all }
    n \geq 0.
    \end{equation}
\end{restatable}

\begin{toappendix}
   
\subsection{Proof of Lemma~\ref{lemma:pumpMatrixS}}\label{subsection:technical}

\keyLemma*
\noindent
The proof %
relies on two lemmas. \cref{lemma:blockDecomp} shows that every p-triangular matrix can be transformed (via an appropriate change of basis) into
a diagonal block matrix where each block is an upper triangular matrix
whose diagonal entries share the same value. It is a standard construction
in linear algebra: for instance, such a change of basis is used to transform
a matrix into its Jordan normal form.

\begin{lemma}\label{lemma:blockDecomp}
    Let $M$ be a p-triangular matrix with diagonal entries
    $\{d_1,d_2,\ldots,d_k\} \subseteq \mathbb{N}$.
    There exists an invertible matrix $P$ 
    such that
    $P^{-1} M P = B_1 \oplus B_2 \oplus \cdots \oplus B_k$
    is a block diagonal matrix
    where each $B_i$ is an upper triangular matrix whose
    diagonal entries all equal $d_i$.
\end{lemma}

We use this result to reduce the study of the powers of a matrix
to the study of the powers of its blocks.
We move to \cref{lemma:pumpMatrixId} that shows that these blocks are easy to handle,
using the fact that the diagonal entries of an individual block share the same value.

\begin{lemma}\label{lemma:pumpMatrixId}
    Let $B$ be an $m \times m$ upper triangular matrix
    whose diagonal entries all have the same value
    $d \in \mathbb{N}$,
    and let $\vec{x},\vec{y} \in \mathbb{Q}^m$.
    There exists a polynomial $p$
    such that
    \[
    \vec{x}^T \cdot B^n \cdot \vec{y}=
    d^n \cdot p(n)
    \textup{ for all }
    n \geq m.
    \]
    Moreover, if $d > 0$ and $p$ is a constant polynomial then
    \[
    \vec{x}^T \cdot B^n \cdot \vec{y}=
    d^n \cdot \vec{x}^T \cdot \vec{y}
     \textup{ for all }
    n \geq 0.
    \]
\end{lemma}

\begin{proof}[Proof of Lemma~\ref{lemma:pumpMatrixId}]
Let $B$ be an upper triangular $m \times m$ matrix
whose diagonal entries all have the same value $d \in \mathbb{N}$.
Remark that in the specific case where $d=0$
the matrix $B$ is nilpotent ($B^m = 0$),
which immediately implies the desired statement:
for all vectors $\vec{x},\vec{y} \in \mathbb{Q}^m$
we have $\vec{x}^T \cdot B^n \cdot \vec{y}= 0$ for all $n \geq m$,
hence every polynomial $p$ fits.
We now suppose that $d>0$.
We decompose $B$ as the sum
of the matrix $d\cdot \id$ obtained by multiplying the identity by $d$,
and the remainder $B_0= B - d\cdot \id$.
This decomposition has the following properties:
\begin{itemize}
\item The two components commute multiplicatively:
$(d \cdot \id) \cdot B_0 = d B_0 = B_0 \cdot (d \cdot \id)$;
\item $B_0$ is a
strictly upper triangular matrix, and it is nilpotent: $B_0^m = 0$.
\end{itemize}
This allows us to give a nice expression for the value
of the powers of $B$:
\[
B^n
= \sum_{i=0}^n \binom{n}{i}(d \cdot \id)^{n-i} \cdot B_0^i
= \sum_{i=0}^{m'} \binom{n}{i}d^{n-i} B_0^i
= d^n \cdot \sum_{i=0}^{m'} \binom{n}{i} \cdot
\frac{B_0^i}{d^i}
\textup{ with }
m' = \min(n,m).
\]
Therefore, for every pair of vectors
$\vec{x},\vec{y} \in \mathbb{Q}^m$
we get
\begin{equation}\label{equation:power}
\vec{x}^T \cdot B^n \cdot \vec{y} =
d^n \cdot \underbrace{\sum_{i=0}^{m'} \binom{n}{i}
\cdot \frac{\vec{x}^T \cdot B_0^i \cdot \vec{y}}{d^i}}_{p(n)}
\textup{ with }
m' = \min(n,m).
\end{equation}
As indicated by the underbrace,
we let $p$ denote the sum in Equation~(\ref{equation:power}) in the case where $m' = m$.
To conclude the proof, it remains to show that the
expression $p$ is a polynomial in~$n$,
and that if $p$ is of degree $0$
then $\vec{x}^T \cdot B^n \cdot \vec{y}=
    d^n \cdot \vec{x}^T \cdot \vec{y}
     \textup{ for all }
    n \geq 0$.
Let us focus on the summands composing $p$:
For every $0 \leq i \leq m$ we have
\[
\binom{n}{i} \cdot \frac{\vec{x}^T \cdot B_0^i \cdot \vec{y}}{d^i} =
\underbrace{\frac{\vec{x}^T \cdot B_0^i \cdot \vec{y}}{d^i
\cdot i (i-1) (i-2) \cdot \ldots \cdot 2 \cdot 1}}_{\in \mathbb{Q}} \cdot
n(n-1)(n-2) \ldots (n-i+1).
\]
This expression is a polynomial of degree $i$
if $\vec{x}^T \cdot B_0^i \cdot \vec{y} \neq 0$,
and it is constantly zero if $\vec{x}^T \cdot B_0^i \cdot \vec{y} = 0$.
Since $p(n)$ is equal to the sum of all these expressions,
we get that $p(n)$ is a polynomial whose degree is equal to the largest $i$ satisfying
$\vec{x}^T \cdot B_0^i \cdot \vec{y} \neq 0$ (or $p(n)$ is constantly zero if no such $i$ exists).
In particular, if $p$ is a constant polynomial
we get that $\vec{x}^T \cdot B_0^i \cdot \vec{y} = 0$ for every $1 \leq i \leq m$, thus Equation~(\ref{equation:power})
yields the desired expression:
\[
\vec{x}^T \cdot B^n \cdot \vec{y}=
    d^n \cdot \binom{n}{0} \cdot \vec{x}^T \cdot B_0^0 \cdot \vec{y}=
    d^n \cdot \vec{x}^T \cdot \vec{y}
     \textup{ for all }
    n \geq 0.
    \qedhere
\]
\end{proof}

\begin{proof}[Proof of Lemma~\ref{lemma:pumpMatrixS}]
Let $M$ be an $m \times m$ upper triangular matrix with diagonal entries $\{d_1,d_2,\ldots,d_n\} \subseteq \mathbb{N}$.
We begin by transforming $M$ into a block diagonal matrix according to Lemma~\ref{lemma:blockDecomp}:
$P^{-1} M P = B_1 \oplus B_2 \oplus \cdots \oplus B_k$.
Since $M^n = P \cdot (B_1^n\oplus B_2^n \oplus \cdots \oplus B_k^n) \cdot P^{-1}$,
we get
\begin{equation}\label{eqn:first}
    \vec{x}^T \cdot M^n \cdot \vec{y}\underrel{\ref{eqn:first}.1}{=}
    \sum_{i=1}^k \vec{x}_i^T \cdot B_i^n \cdot \vec{y}_i\underrel{\ref{eqn:first}.2}{=}
    \sum_{i=1}^k d_i^n \cdot p_i(n) \textup{ for all } n \geq m, \textup{ where: }
\end{equation}
\begin{description}
\item[\ref{eqn:first}.1]
The vectors $\vec{x}_1^T,\vec{x}_2^T, \ldots, \vec{x}_k^T$, respectively $\vec{y}_1,\vec{y}_2, \ldots, \vec{y}_k$,
are obtained by cutting $\vec{x}^T \cdot P$, respectively $P^{-1} \cdot \vec{y}$,
into parts whose sizes fit the blocks $B_1,B_2,\ldots,B_k$;
\item[\ref{eqn:first}.2]
The polynomials $p_1,p_2,\ldots,p_n$ are obtained by applying Lemma~\ref{lemma:pumpMatrixId} to each block.
\end{description}
This concludes the proof of Equation~(\ref{equation:sum}) of Lemma~\ref{lemma:pumpMatrixS}.

In order to prove the second part, we now suppose that the matrix $M$ is invertible.
Since $M$ is an upper triangular matrix, this implies that none of its diagonal entries $d_i$ equals zero.
Therefore, if on top of that we suppose that $p_1,p_2,\ldots,p_k$ are all constant polynomials,
and that only $p_j$ is not constantly $0$,
we can apply the second part of Lemma~\ref{lemma:pumpMatrixId} to get:
\begin{description}
\item[\ref{eqn:sec}.1]
$\vec{x}_j^T \cdot B_j^n \cdot \vec{y}_j = d_j^n \vec{x}_i^T \cdot \vec{y}_i$ for all $n \in \mathbb{N}$.;
\item[\ref{eqn:sec}.2]
For every $i \neq j$, $\vec{x}_i^T \cdot B_i^n \cdot \vec{y}_i = 0$
for all $n \in \mathbb{N}$ (thus in particular $\vec{x}_i^T \cdot \vec{y}_i = 0$).
\end{description}
We use these properties to refine
Equation~\ref{eqn:first} into Equation~(\ref{equation:invertible}) of Lemma~\ref{lemma:pumpMatrixS}:
\begin{equation}\label{eqn:sec}
    \vec{x}^T \cdot M^n \cdot \vec{y}=
    \sum_{i=1}^k \vec{x}_i^T \cdot B_i^n \cdot \vec{y}_i\underrel{\ref{eqn:sec}.1-\ref{eqn:sec}.2}{=}
    d_j^n \cdot \vec{x}_j^T \cdot \vec{y}_j\underrel{\ref{eqn:sec}.2}{=}
    \sum_{i=1}^k d_j^n \cdot \vec{x}_i^T \cdot \vec{y}_i=
    d_j^n \cdot \vec{x}^T \cdot \vec{y} \textup{ for all } n \in \mathbb{N}.\qedhere
\end{equation}
\end{proof}
 
\end{toappendix}

It is standard that a sequence $\vec{x}^T \cdot M^n \cdot \vec{y}$ forms a linear recurrence sequence which admits a closed form like the one given in~\cref{equation:sum}.
Thus most of \cref{lemma:pumpMatrixS} can be derived from standard results on linear recursive sequences (see \eg~\cite[Section 2]{halava2005skolem}). We prove it in \cref{subsection:technical} to get the precise form required in~\cref{equation:invertible}. It is nontrivial as it holds for $n \ge 0$ (not just $n$ big enough). This will be crucial later in the proofs.
Overall, \cref{lemma:pumpMatrixS} has two purposes:

\cref{equation:sum}
expresses the behaviours of $\A$ over families of the form
$(uv^nw)_{n \geq |\A|}$.
This allows us to define the notion of \emph{pumpability}
(resp. {blind pumpability})
by forbidding \emph{bad} behaviours
(the ones that match no unambiguous automaton).
We show that this is a necessary criterion
for unambiguisability (resp. determinisability)
(see \cref{prop:UnambIsPump}).

\cref{equation:invertible}
allows us to establish pumpability (resp. {blind pumpability})
is sufficient to ensure unambiguisability
(resp. determinisability) for polynomially-ambiguous WA (see \cref{prop:pumpIsUnamb}).

One of the reasons we restrict ourselves 
to p-triangular matrices in \cref{lemma:pumpMatrixS}
(and further definitions) is to avoid
dealing with complex numbers.
For a general matrix $M$ an expression similar to \cref{equation:sum} also holds,
with the $d_i$ being \emph{eigenvalues} of $M$
(which happen to be the diagonal entries
if $M$ is p-triangular).
Eigenvalues of a matrix over $\mathbb{N}$
do not have to be rational numbers,
which make it more difficult to distinguish good
behaviours from bad.

\subparagraph{Excluding bad behaviours}
As a direct consequence of Equation~(\ref{equation:sum}),
we get that for every $u,v,w \in \Sigma^*$ such that $M(v)$ is a p-triangular matrix,
\[
    \A(uv^{n}w) = \sum_{i=1}^k d_i^{n} \cdot p_i(n)
    \textup{ for every } n \geq |\A|,
\]
where $\{d_1,d_2,\ldots,d_n\}$ is the set of diagonal entries of $M(v)$,
and $p_1,p_2,\ldots,p_n$ are polynomials.
We show that if $\A$ is unambiguisable,
then this expression
needs to collapse to some simple periodic expression
for every choice of $u,v,w$.
We formalise this through the notions of \emph{pumpable automata}
and \emph{blindly pumpable automata}
(in the latter, the pumping does not depend on the suffix,
which, as we will show, is required for determinisability).

\begin{definition}\label{def:pumpable}
    A weighted automaton $\A$ 
    is \emph{pumpable}
    if for all $u,v,w \in \Sigma^*$ such that $M(v)$
    is a p-triangular matrix
    there is an entry $d$ of the diagonal of $M(v)$ 
    satisfying
    \[
    \A(uv^{|\A| + n}w) = d^{n} \cdot \A(uv^{|\A|}w)
    \textup{ for every }
    n \in \mathbb{N}.
    \]
    Moreover, we say that $\A$ is \emph{blindly pumpable}
    if the entry $d$ does not depend on the suffix $w$,
    that is, for all $u,v \in \Sigma^*$ such that $M(v)$
    is a p-triangular matrix
    there is an entry $d$ of the diagonal of $M(v)$ 
    satisfying
    \[
    \A(uv^{|\A| + n}w) = d^{n} \cdot \A(uv^{|\A|}w)
    \textup{ for every } w \in \Sigma^*
    \textup{ and }
    n \in \mathbb{N}.
    \]
\end{definition}
The interest of this definition is reflected by Proposition~\ref{prop:UnambIsPump}, proved in Section~\ref{section:BAD}.

\begin{restatable}{proposition}{UnambIsPump}\label{prop:UnambIsPump}
Every unambiguisable weighted automaton is pumpable,
and every determinisable weighted automaton is blindly pumpable.
\end{restatable}

\subparagraph{Taking advantage of good behaviours}
Unfortunately, the converse statement
of Proposition~\ref{prop:UnambIsPump}
does not hold in general, because of the following limitations:
\begin{enumerate}
    \item Pumpability refers to p-triangular matrices,
    which do not appear in some automata;
    \item Pumpability only guarantees a periodic behaviour
    over the words where some factor $v$ is pumped several times:
    How do we ensure that this nice behaviour extends to
    words where such a repetition never happens,
    for instance square-free words?
\end{enumerate}
To overcome these limitations, we %
restrict ourselves to \emph{polynomially-ambiguous} automata.
In this setting matrices with an idempotent structure are p-triangular matrices up to a rearrangement of the states (Lemma~\ref{lemma:polyAmbTriangular}),
and the presence of such matrices is guaranteed by Ramsey's theorem (Lemma~\ref{lemma:Ramsey}).
On top of this, we remark that Equation~(\ref{equation:invertible}) in Lemma~\ref{lemma:pumpMatrixS} would counteract the second limitation, but the periodic behaviour starts
from the first appearance only for invertible matrices. The restriction to polynomially-ambiguous WA further allows us to show that sufficiently long words contain invertible-like idempotent entries in which Equation~(\ref{equation:invertible}) can be applied. 

Thus, restricting ourselves to polynomially-ambiguous automata
allows us to get rid of the limitations preventing
the converse of Proposition~\ref{prop:UnambIsPump}. We obtain the following result,
proved in Section~\ref{section:GOOD}:

\begin{restatable}{proposition}{pumpIsUnamb}\label{prop:pumpIsUnamb}
Let $\A$ be a polynomially-ambiguous automaton.
If $\A$ is pumpable then it is unambiguisable,
and if $\A$ is blindly pumpable then it is determinisable.
\end{restatable}

\subparagraph{Deciding Pumpability}
With Propositions~\ref{prop:UnambIsPump} and~\ref{prop:pumpIsUnamb},
we have identified a class of weighted automata
for which pumpability, respectively blind pumpability,
characterises unambiguisability, respectively determinisability.
We finally show in Section~\ref{section:decision}
that we can decide the two former notions
in polynomial space:

\begin{restatable}{proposition}{propPumpIsDecidable}\label{prop:pumpIsDecidable}
    We can decide in polynomial space
    whether a given weighted automaton is pumpable,
    respectively blindly pumpable.
\end{restatable}

\noindent
Together, \cref{prop:pumpIsUnamb,prop:UnambIsPump,prop:pumpIsDecidable} entail our main theorem.

\maintheorem*

\section{Pumpable automata}\label{section:BAD}

The goal of this section is to prove~\cref{prop:UnambIsPump}:
\UnambIsPump*
We begin by establishing the possible behaviours of unambiguous automata, respectively deterministic automata,
over families of words of the form $(uv^nw)_{n \in \mathbb{N}}$ (Lemma~\ref{lemma:pumpUnambiguous}).
Then, to prove Proposition~\ref{prop:UnambIsPump},
we show that, given a weighted automaton $\A$,
for each triple $u,v,w \in \Sigma^*$,
if we take the possible values of $(\A(uv^nw))_{n \in \mathbb{N}}$
(which are described by Equation~(\ref{equation:sum}) of Lemma~\ref{lemma:pumpMatrixS})
and remove all the behaviours that
do not fit an unambiguous (resp. deterministic) weighted automaton
(which are described by Lemma~\ref{lemma:pumpUnambiguous}),
then we are left with exactly the behaviours fitting a pumpable (respectively blindly pumpable) automaton.

\begin{lemma}\label{lemma:pumpUnambiguous}
Let $\A$ be a weighted automaton.
\begin{enumerate}[nolistsep]
\item
If $\A$ is unambiguous,
then for all $u,v,w \in \Sigma^*$
there exist $m,d \in \mathbb{N}$ and $\lambda > 0$ such that
\[
\A(uv^{m + \lambda n}w) = d^n \cdot \A(uv^{m}w)
\textup{ for all }
n \in \mathbb{N}.
\]
\item
If $\A$ is deterministic,
then for all $u,v \in \Sigma^*$
there exist $m,d \in \mathbb{N}$ and $\lambda > 0$ such that
\[
\A(uv^{m + \lambda n}w) = d^n \cdot \A(uv^{m}w)
\textup{ for all }
w \in \Sigma^*
\textup{ and }
n \in \mathbb{N}.
\]
\end{enumerate}
\end{lemma}

\begin{proof}
These results follow from a standard pumping argument:
in every automaton, long enough runs eventually
visit a cycle.
Since unambiguous weighted automata have
at most one run over each input word,
pumping such a cycle amounts to multiplying the weight of the cycle by a constant.
On top of this, the second item follows from the fact that
in a deterministic automaton the pumped constant
cannot depend on the suffix that is yet to be read.
Let us now prove these results formally.

\subparagraph{1.}
    Let $\A$ be an unambiguous automaton,
    and let $u,v,w \in \Sigma^*$.
    First, remark that if for every $m \geq |\A|$
    the automaton $\A$ has no run over $uv^{m}w$,
    then $\A(uv^{n}w) = 0$,
    and the statement is easily satisfied by picking
    $m= |\A|$, $d = 0$ and $\lambda = 1$.
    Now let us suppose that there exists $m \geq |\A|$
    such that $\A$ has a run $\rho$ over $uv^{m}w$.
    Since $|\A|$ denotes the number of states of $\A$,
    $\rho$ can be factorised into three subruns
    such that the middle part is a cycle
    which processes a non-empty power $v^{\lambda}$ of the word $v$.
    Therefore, if we denote by $d$ the weight of this middle part,
    we get that for all $n \in \mathbb{N}$
    the automaton $\A$ has a run over $uv^{m + \lambda n}w$
    which has weight $d^n \cdot \A(uv^{m}w)$.
    However, since $\A$ is unambiguous by supposition,
    this is the only run over this input word,
    hence $\A(uv^{m + \lambda n}w) = d^n \cdot \A(uv^{m}w)$,
    which concludes the proof.

 \subparagraph{2.}
    Let $\A$ be a deterministic automaton,
    and let $u,v \in \Sigma^*$.
    The proof is nearly identical to the one for the unambiguous setting:
    If for every $m \geq |\A|$
    the automaton $\A$ has no path over $uv^{n}w$
    starting in an initial state,
    then $\A(uv^{n}w) = 0$, and we are done.
    If there exists $m \geq |\A|$ such that
    $\A$ has a path over
    $uv^{m}w$ starting from the initial state,
    then this path will visit a cycle while reading a non-empty power
    $v^{\lambda}$ of $v$, hence
    for every $n \in \mathbb{N}$
    there exists a path in $\A$
    over $uv^{m + \lambda n}$ that starts in an initial state
    and has weight $d^n \cdot \A(uv^{m}w)$.
    Since $\A$ is deterministic, there cannot be any other such path,
    hence for every $w \in \Sigma^*$ we get that
    $\A(uv^{m + \lambda n}w) = d^n \cdot \A(uv^{m}w)$,
    as required.
\end{proof}

\begin{proof}[Proof of \cref{prop:UnambIsPump}]
Fix a weighted automaton $\A = (Q, \Sigma, M, I, F)$.

\subparagraph{Unambiguisability implies pumpability}
Let us suppose that $\A$ is unambiguisable, and let $\U$ be an unambiguous automaton equivalent to $\A$.
Let $u,v,w \in \Sigma^*$ such that $M(v)$ is a p-triangular matrix.
We compare the behaviour of $\A$ and $\U$ over the family of words $(uv^nw)_{n \in \mathbb{N}}$.
The behaviour of $\A$ is described by 
Equation~(\ref{equation:sum}) of Lemma~\ref{lemma:pumpMatrixS}:
\begin{equation}\label{equation:ABehaviour}
    \A(uv^{n}w) = \sum_{i=1}^k d_i^{n} \cdot p_i(n)
    \textup{ for every } n \geq |\A|,
\end{equation}
where $\{d_1,d_2,\ldots,d_n\}$ is the set of diagonal entries of $M(v)$
and $p_1,p_2,\ldots,p_n$ are polynomials.
The behaviour of $\U$ is described by  Lemma~\ref{lemma:pumpUnambiguous}:
There exist $d,m \in \mathbb{N}$ and $\lambda > 0$ such that
\[
\U(uv^{m + \lambda n}w) = d^n \cdot \U(uv^{m}w)
\textup{ for all }
n \in \mathbb{N}.
\]
Since $\A$ is equivalent to $\U$, these two behaviours need to match.
Intuitively, this implies that Equation~(\ref{equation:ABehaviour}) collapses
to a single term of the form $\A(uv^{n}w) = \delta^{n} \cdot C$ (so that it mirrors the behaviour of $\U$),
which proves the desired result as this fits the definition of a pumpability.
We now present the formal details proving this intuition.
First, remark that
\[
\left(\sum_{i=1}^k d_i^{m + \lambda n} \cdot p_i(m + \lambda n)\right)
- d^n \cdot \U(uv^{m}w) = 
\A(uv^{m + \lambda n}w)-\U(uv^{m + \lambda n}w) = 0
\textup{ for all } n \geq |\A|.
\]
Let us rewrite the left-hand side as
$\sum_{i=1}^k (d_i^\lambda)^n \cdot p'_i(n)$,
where 
\[
p'_i(n) = \left\{ 
\begin{array}{lll}
d_i^{m} p_i({m} + \lambda n) & \textup{if } d_i^\lambda \neq d;\\
d_i^{m} p_i({m} + \lambda n) - \U(uv^{{m}}w) & \textup{if } d_i^\lambda = d.
\end{array} \right.
\]
Since $\sum_{i=1}^k (d_i^\lambda)^n \cdot p'_i(n) =0$ for all $n \in \mathbb{N}$
and the $d_i$ are distinct positive integers, every $p_i'$ is constantly $0$:
Otherwise, $\sum_{i=1}^k (d_i^\lambda)^n \cdot p'_i(n)$ would behave asymptotically like
$(d_j^\lambda)^n \cdot p'_j(n)$ (contradicting the fact that it is $0$),
where $d_j$ is the largest $d_i$ such that the corresponding $p'_i$ is not constantly~$0$.
As a consequence, for every $1 \leq i \leq k$:
\begin{itemize}[nolistsep]
\item
If $d_i^\lambda \neq d$, then the fact that
$p_i'(n) = d_i^{m} p_i(m + \lambda n)$ is constantly $0$
implies that either $d_i = 0$ or $p_i(n)$ is constantly $0$;
\item
If $d_i^\lambda = d$, then the fact that 
$p_i'(n) = d_i^{m} p_i(m + \lambda n) - \U(uv^{m}w)$
is constantly $0$ implies that either $d_i = 0$ or $p_i(n)$ is constantly $\frac{\U(uv^{m}w)}{d_i^m}$.
\end{itemize}
We can update Equation~(\ref{equation:ABehaviour}) accordingly:
If there exist $d_j$ such that $d = d_j^\lambda$,
then $\A(uv^{n}w) = d_j^n \cdot \frac{\U(uv^{m}w)}{d_i^m}$ for every $n \geq |\A|$.
Otherwise, $A(uv^{n}w) = 0$ for every $n \geq |\A|$.
Remark that both cases fit the requirement of the definition of pumpability:
$\A(uv^{|\A|+n}w) = d_j^n \cdot \A(uv^{|\A|}w)$.
Since this holds for every triple $u,v,w \in \Sigma^*$ required in \cref{def:pumpable}
we deduce that $\A$ is pumpable.

\subparagraph{Determinisability implies blind pumpability}
If $\A$ is equivalent to a deterministic automaton $\D$,
in particular $\A$ is unambiguisable,
thus, as we just proved, it is pumpable.
Towards building a contradiction, suppose that $\A$ is not blindly pumpable:
there exists $u,v,w_1,w_2 \in \Sigma^*$ such that
$\A(uv^{|\A|+n}w_1) = d_1^n \cdot \A(uv^{|\A|}w_1)$
and
$\A(uv^{|\A|+n}w_2) = d_2^n \cdot \A(uv^{|\A|}w_2)$ for all $n \in \mathbb{N}$,
yet $d_1 \neq d_2$ and $\A(uv^{|\A|}w_1) \neq 0 \neq \A(uv^{|\A|}w_2)$.
We compare these expression with the behaviours of $\D$, 
described by Lemma~\ref{lemma:pumpUnambiguous}:
There exist $d,m\in \mathbb{N}$ and
$\lambda>0$ such that $\D(uv^{m + \lambda n}w_1) = d^n \cdot \D(uv^{m}w_1)$
and $\D(uv^{m + \lambda n}w_1) = d^n \cdot \D(uv^{m}w_1)$ for all $n \in \mathbb{N}$.
Since $d_1$ and $d_2$ are distinct integers,
$d$ cannot match both $d_1^\lambda$ and $d_2^\lambda$, 
hence $\A$ and $\D$ disagree on the value of either
$uv^{m + \lambda n}w_1$ or $uv^{m + \lambda n}w_2$ for $n$ sufficiently large.
This contradicts the fact that the two automata are equivalent.
\end{proof}

\section{Depumpable automata}\label{section:GOOD}

This section is devoted to the proof of~\cref{prop:pumpIsUnamb}, which we recall:
\pumpIsUnamb*
Our proof relies on the following
lemma, which shows that, once we restrict ourselves to 
polynomially-ambiguous automata,
pumpable automata are also \emph{depumpable}:
\begin{lemma}[Depumping lemma]\label{lemma:depump}
    Let $\A$ be a polynomially-ambiguous automaton.
    There exists a constant $\ramsey$
    such every word $u \in \Sigma^*$ satisfying $|u| = \ramsey$
    can be decomposed into three parts
    $u = u_1u_2u_3$ such that $u_2 \neq \epsilon$
    and
    \begin{enumerate}[nolistsep]
        \item\label{item:Pumpable}
        If $\A$ is pumpable,
        for each $v \in \Sigma^*$ there exist an entry $d$
        of the matrix $M(u_2)$ such that $\A(uv) = d \cdot \A(u_1u_3v)$;
        \item\label{item:blindlyPumpable}
        If $\A$ is blindly pumpable, there exists and entry $d$
        of the matrix $M(u_2)$ such that for all $v \in \Sigma^*$
        we have $\A(uv) = d \cdot \A(u_1u_3v)$.
    \end{enumerate}  
\end{lemma}
Remark that the order of the quantifiers is different in 
the two items:
In Item~\ref{item:blindlyPumpable} the entry $d$
does not depend on the suffix $v$.
The proof of Lemma~\ref{lemma:depump}
is presented in Subsection~\ref{subsection:depump}.

\begin{proof}[Proof of Proposition~\ref{prop:pumpIsUnamb}]
Let $\A$
be a polynomially-ambiguous automaton.
We prove both parts of Proposition~\ref{prop:pumpIsUnamb}
independently.
Remark that we only use the assumption
that $\A$ is polynomially-ambiguous
and pumpable (respectively blindly pumpable)
to enable the use of the depumping lemma:
In other words, extending the depumping lemma
to a wider class of automata would result in extending
Proposition~\ref{prop:pumpIsUnamb} in the same manner.

\subparagraph{Pumpability implies unambiguisability}
Let us suppose that $\A$ is pumpable.
We prove that it is unambiguisable relying on
the following characterisation
by Bell and Smertnig. 

\begin{theorem}[\cite{bell2021noncommutative}]\label{theorem:Bell}
A weighted automaton over integers $\A$ is unambiguisable if and only if
the set of prime divisors of the set 
$\{\A(w) \mid w\in\Sigma^*\}$ is finite. 
\end{theorem}

We show that the set of prime divisors of the weights of $\A$
is finite by building an upper bound according to the following intuition:
For each $u \in \Sigma^*$, repeated application of
the depumping lemma allows us to express $\A(u)$
as a product of the form $\A(u') \cdot \prod_{i=1}^n d_i$,
where both $\A(u')$ and the $d_i$ are bounded.
This gives us a bound for the prime divisors of $\A(u)$. 

Let us formalise this idea.
Let $N \in \mathbb{N}$ be an integer such that
for every $u \in \Sigma^*$ shorter than $\ramsey$
all the entries of $M(u)$ are smaller than $N$
and $\A(u) < N$.
We prove by induction on the length that for every word $u \in \Sigma^*$,
the prime factors of $\A(u)$ are all smaller than $N$.

The base case of the induction is immediate:
If $|u| < \ramsey$, then by definition of $N$ we get that $A(u)$ is smaller than $N$,
thus its prime factors are also smaller than $N$.

Now let us suppose that $|u| \geq \ramsey$ and that for every word $u'$ shorter than $u$
the prime factors of $\A(u')$ are smaller than $N$.
By applying Lemma \ref{lemma:depump} to the prefix of size $\ramsey$
of $u$ we obtain a decomposition $u =  u_1u_2u_3$
satisfying $0 < |u_2| \leq \ramsey$ and 
 $\A(u) = d \cdot \A(u_1u_3)$ for some entry $d$ of $M(u_2)$.
We conclude by remarking that $d < N$
since $|u_2| < \ramsey$,
and the prime factors of $\A(u_1u_3)$
are smaller than $N$ by the induction hypothesis.

\subparagraph{Blind pumpability implies determinisability}
Let us suppose that $\A$ is blindly pumpable.
We use the depumping Lemma
to build a deterministic automaton
$\D$ equivalent to $\A$.

We begin with an informal description of the behaviour of $\D$.
The states of $\D$ are the words $u \in \Sigma^*$
satisfying $|u| \leq \ramsey$.
The automaton $\D$
starts by keeping track in its state of the input word
read so far, and preserves a weight equal to $1$.
Whenever $\D$ reaches a state corresponding
to a word $u$ of length $\ramsey$,
it ``depumps'' it:
According to Lemma~\ref{lemma:depump} there exists
a decomposition $u = u_1u_2u_3$ and an entry
$d$ of $M(u)$
such that for every suffix $w \in \Sigma^*$
$\A(uw) = d \cdot \A(u_1u_3w)$.
Therefore, $\D$ can immediately produce the
weight $d$ corresponding to the infix $u_2$,
and transition towards the state
$u_1u_3$.
Finally, when $\D$ finishes reading its input,
it produces the weight $\A(v)$
corresponding to its current state $v$.
In order to prove that such an automaton
is indeed equivalent to $\A$, we now formalise this
construction.

We define the deterministic automaton
$\D = (Q', \Sigma, M', I', F')$ as follows.
First, 
we introduce the depumping function
$\textsf{cut}: \Sigma^{\ramsey} \rightarrow \Sigma^{*} \times \mathbb{Z}$
that \emph{depumps deterministically}
the words of size $\ramsey$:
it picks, for each $u \in \Sigma^*$ such that $|u| = \ramsey$,
a pair $(u',d)$
satisfying
\begin{itemize}[nolistsep]
    \item $|u'| < \ramsey$;
    \item $d$ is an entry of $M(u')$;
    \item $\A(uv) = d \cdot \A(u'v)$ for every $v \in \Sigma^*$.
\end{itemize}
Note that this function exists:
Lemma~\ref{lemma:depump} guarantees the existence of
at least one such pair $(u',d)$ for every $u \in \Sigma^{\ramsey}$. To ensure that $\D$ is deterministic  $\textsf{cut}$ can pick any fixed pair for every $u\in \Sigma^{\ramsey}$, for instance, the minimal pair according to the lexicographical order.
We now define the components of $\D$:

\begin{itemize}[nolistsep]
    \item $Q' = \{q_u \in \Sigma^* \mid |u| < \ramsey\}$;
    \item For every $u \in Q'$ and $a \in \Sigma$,
    \begin{itemize}[nolistsep]
        \item
        If $|ua| < \ramsey$ then $q_u \xrightarrow{a:1} q_{ua}$;
        \item
        If $|ua| =\ramsey$ then
        $q_{u} \xrightarrow{a:d} q_{u'}$
        where $\textsf{cut}(ua) = (u',d)$.
    \end{itemize}
    \item $I'(q_{\epsilon}) = 1$ and $I'(q_{u}) = 0$
    for every $u \neq \epsilon$;
    \item $F'(q_{u}) = \A(u)$.
\end{itemize}
The automaton $\D$ is deterministic:
There is exactly one initial state,
and for every state $q_{u}$ and letter $a$
there is exactly
one outgoing transition from $u$ labelled by $a$.
All that remains to show is that $\D(u) = \A(u)$
for every $u \in \Sigma^*$.
To this end, we prove the following claim:
\begin{claim}
The weight $x$ of the (single) path
$q_{\epsilon},q_{u_1},q_{u_2},\dots,q_{u_n}$
of $\D$ over $u$
starting from the initial state $q_\epsilon$
satisfies $\A(uw) = x \cdot \A(u_nw)$
for every $w \in \Sigma^*$.
\end{claim}
Remark that, by definition of the final vector $F$,
this immediately implies that $\D(u) = \A(u)$.
We prove the claim by induction on the size of $u$.
If $|u| \leq \ramsey$
the result is immediate:
the path of $\D$ over $u$ that starts from the initial state
has weight $1$ and ends in the state $q_u$ (thus $u_n = u$).
Now suppose that $|u| > \ramsey$.
Let us denote $u = va$ with
$v \in \Sigma^*$ and $a \in \Sigma$.
Remark that the path
of $\D$ over $v$ starting from the initial state
is obtained by removing the last transition
$q_{u_{n-1}} \xrightarrow{a:d} q_{u_n}$
of the path over $u$.
Therefore, for every $w \in \Sigma^*$,
applying the induction hypothesis to $v$
yields that
$\A(uw) = \A(vaw) = \frac{x}{d} \cdot \A(u_{n-1}aw)$.
We distinguish two possibilities,
depending on the type of transition used while reading the last letter
$a$ of $u$:
\begin{itemize}
    \item
    If $|u_{n-1}a| < \ramsey$,
    then $u_n = u_{n-1}a$ and $d = 1$,
    hence
    $\A(u) =
    \frac{x}{d} \cdot \A(u_{n-1}aw) =
    x \cdot \A(u_{n}w)$;
    \item
    If $|u_{n-1}a| = \ramsey$,
    then $\textsf{cut}(u_{n-1}a) = (u_n,d)$,
    hence
    $\A(u) =
    \frac{x}{d} \cdot \A(u_{n-1}aw) =
    x \cdot \A(u_{n}w)$.
\end{itemize}
The last equality follows from the definition of $\textsf{cut}$.
\end{proof}

\begin{remark}\label{remark:deterministic}
While our decision procedure is in PSPACE, the size of the automaton $\D$, should it be constructed, depends on $\ramsey$. We will see in the next section this will be quite big (non-elementary). Our decision procedure does not need to construct $\D$.
\end{remark}

\subsection{Replacing pumpable idempotents with invertible matrices}

\newcommand{\tower}[1]{\textsf{tower}_{#1}}

To prove \cref{lemma:depump} we will need a technical lemma that allows us to find pumpable fragments and replace them with invertible matrices. Some ideas here are similar as in~\cite{BumpusHKST20}, in particular the idea behind \cref{claim:inv} is similar to \cite[Lemma 8]{BumpusHKST20}.

Let us consider the family of function
$(\tower{r})_{r \in \mathbb{N}}$
defined inductively as follows:
\[
\begin{array}{ll}
\tower{0}(x) = x &
\textup{ for all } x \in \mathbb{N};\\
\tower{r+1}(x) = x \cdot \tower{r}(x^x) &
\textup{ for all } x \in \mathbb{N}.
\end{array}
\]

\begin{lemma}\label{lem:tower}
    Let $r$ be the maximal rank of the matrices
    $\{M(a) \mid a \in \Sigma\}$.
    Then for all $\ell \geq 1$,
    every word $u^* \in \Sigma^*$
    satisfying $|u| \geq \tower{r}(2\ell \cdot |\Sigma|)$
    can be decomposed as
    $u = u_{\vdash} u_1 u_2 \ldots u_\ell u_{\dashv}$
    such that
    \begin{itemize}
        \item 
        for every $1 \leq i \leq \ell$
        the word $u_i$ is nonempty;
        \item 
        for every $1 \leq i \leq j \leq \ell$
        there exists an invertible matrix $M$ satisfying,
        for all $n \geq 0$:
        \begin{equation}\label{eq:pump}
        M(u_{\vdash} u_1 u_2 \ldots u_{i-1}
        (u_i \ldots u_{j})^n u_{j+1} \ldots u_\ell u_{\dashv})
        = 
        M(u_{\vdash} u_1 u_2 \ldots u_{i-1})
        \cdot 
        M^n
        \cdot 
        M(u_{j+1} \ldots u_\ell u_{\dashv}).
        \end{equation}
    \end{itemize}
    
\end{lemma}

The proof  will rely on the following claim.
    \begin{claim}\label{claim:inv}
        If the rank of $ABA$ is equal to the rank of $A$,
        there exists an invertible matrix $C$
        satisfying $A(BA)^n = AC^n$ for all $n \geq 0$. 
    \end{claim}
    \begin{claimproof}
    Since the rank of $A$ and $ABA$ are equal we conclude that their images are also the same. Let $V \subseteq \Q^m$ be this image. Note that $BA$ can be seen as a linear transformation on $\Q^m$, and $BA$ restricted to $V$, \ie $BA|_V$, is a bijection. Let $\vec{v}_1,\ldots,\vec{v}_s$ be a basis of $V$ and let $\vec{v}_{s+1}, \ldots, \vec{v}_m$ be such that $\vec{v}_1,\ldots,\vec{v}_m$ is a basis of $\Q^m$. 
    
    Recall that a linear transformation $\Q^m \to \Q^m$ is uniquely defined by fixing the images of a basis. Let $\vec{w}_1,\ldots, \vec{w}_s$ be the images of $\vec{v}_1,\ldots, \vec{v}_s$ when applying the linear transformation $BA$, and observe that these also span $V$. The matrix $C$ is defined by the linear transformation that maps $\vec{v}_i$ to $\vec{w}_i$ for all $1 \le i \le s$ and $\vec{v}_i$ to $\vec{v}_i$ for $s+1 \le i \le m$. 

    It is clear that $C$ is invertible as its image has $m$ independent vectors. To prove that $A(BA)^n = AC^n$ we prove that both linear transformations are defined in the same way on $\vec{v}_1,\ldots,\vec{v}_m$. Indeed, vectors $\vec{v}_i$ for $s+1 \le m$ are mapped to $\vec{0}$ by both transformations. The remaining vectors have the same image since $BA|_V = C|_V$ and $BA|_V$ is a bijection.
\end{claimproof}

\begin{proof}[Proof of \cref{lem:tower}]
    We prove the lemma by induction on $r$.
    
    If $r = 0$ the proof is straightforward:
    by definition of $r$ for every letter
    $a \in \Sigma$ the matrix $M(a)$
    has rank $0$,
    thus it is the null matrix.
    Then for every word
    $u = a_1a_2 \ldots a_{n}$
    of size $n \geq \tower{0}(2\ell \cdot |\Sigma|)
    = 2\ell \cdot |\Sigma| > \ell \cdot |\Sigma|$,
    the left-hand side of Equation~\eqref{eq:pump}
    is equal to the null matrix for every choice of $i,j$ and $n$.
    Therefore, 
    we satisfy the statement by setting
    $u_{\vdash} = \epsilon$,
    $u_i = a_i$ for every $1 \leq i \leq \ell \cdot |\Sigma|$
    and $u_{\dashv} = a_{\ell \cdot |\Sigma| + 1} \ldots a_n$.

    Now let us suppose that $r>0$ and that the lemma holds
    for $r-1$.
    Let $x = 2\ell \cdot |\Sigma|$,
    and let $u \in \Sigma^*$ be a word
    satisfying
    $|u| > \tower{r}(x) = x \tower{r-1}(x^x)$.
    We consider the decomposition
    $u = v_1v_2 \ldots v_{\tower{r-1}(x^x)}v$
    such that $|v_i| = x$
    for every $1 \leq i \leq \tower{r-1}(x^x)$.
    We distinguish two cases:
    \begin{itemize}
        \item
        Suppose that there
        exists $1 \leq i \leq \tower{r-1}(x^x)$ such that
        the rank of $M(v_i)$ is $r$.
        Then for every infix $v'$ of $v_i$
        the rank of $M(v')$ is also $r$.
        Moreover,
        since $|v_i| = x > \ell \cdot |\Sigma|$,
        there is one letter $a \in \Sigma$
        that occurs $\ell + 1$ times in $v_i$.
        We get the statement of the lemma as a direct
        consequence of Claim~\ref{claim:inv}.
        \item
        Suppose that for every $1 \leq i \leq \tower{r-1}(x^x)$
        the matrix $M(v_i)$ has a rank smaller than $r$.
        In order to apply the induction hypothesis,
        we now consider the alphabet
        $\Gamma = \{ v_i \mid 1 \leq i \leq \tower{r-1}(x^x)\}$.
        Then for every letter $a \in \Gamma$
        we have that the rank of $M(a)$ is smaller than or equal to $r-1$.
        Moreover, since the length of each $v_i$ is $x$,
        we get that $|\Gamma| \leq 
        |\Sigma|^{x}$.
        Therefore, the word
        $w = v_1v_2 \ldots v_{\tower{r-1}(x^x)} \in \Gamma^*$
        satisfies
        \[
        |w| = \tower{r-1}(x^x)
        = \tower{r-1}((2\ell \cdot |\Sigma|)^{x})
        > \tower{r-1}(2\ell \cdot |\Sigma|^{x})
        \geq \tower{r-1}(2\ell \cdot |\Gamma|).
        \]
        As a consequence,
        we can apply the induction hypothesis
        to obtain a decomposition
        of the word $w =v_1v_2 \ldots v_{\tower{r-1}(x^x)} \in \Gamma^*$,
        which can be transferred back to $u$. \qedhere
    \end{itemize}
\end{proof}

\subsection{Proof of depumping lemma (Lemma~\ref{lemma:depump})}\label{subsection:depump}
We need two technical results. The first lemma, intuitively, shows that we can find idempotents in long enough words. It is a direct consequence of~\cite[Theorems~1 and~2]{Jecker21}.\footnote{
The use of~\cite{Jecker21} is enabled thanks to
our \emph{non-negative transitions} assumption (see Remark~\ref{remark:monoid}).
}
The second lemma shows that idempotents for polynomially-ambiguous WA are p-triangular. It is proved in \cref{section:Idempotent}.
\begin{lemma}\label{lemma:Ramsey}
Given a weighted automaton $\A$: let $\ell = (3 \cdot 2^{4|\A|^2})^L$, where $L = \frac{|\A|^2 + |\A| + 2}{2}$.
Let $u = u_1 \ldots u_{\ell} \in \Sigma^+$, where $u_i \in \Sigma^+$ for all $1 \le i \le \ell$. There exist $1 \le i \le j \le \ell$
such that $M(u_i \ldots u_j)$ has idempotent structure.
\end{lemma}

\begin{restatable}{lemma}{polyAmbTriangular}\label{lemma:polyAmbTriangular}
Let $\A$ be a polynomially-ambiguous weighted automaton.
For every $u \in \Sigma^*$,
if $M(u)$ has an idempotent structure then it is p-triangular. 
\end{restatable}

\begin{toappendix}
\section{Proofs for Section~\ref{section:GOOD} (Depumpable automata)}\label{section:Idempotent}
\polyAmbTriangular*
\begin{proof}
Let $\A$ be a polynomially-ambiguous automaton,
let $m$ denote the size of $\A$
and let $u \in \Sigma^*$ such  that $M(u)$ has an idempotent structure.
We construct a permutation $\sigma: \{1,2,\ldots,m\} \rightarrow \{1,2,\ldots,m\}$
such that $P_{\sigma}^{-1} \cdot M(u) \cdot P_{\sigma}$ is an upper triangular matrix,
where $P_{\sigma}$ is defined as
\begin{equation}\label{equation:permutationMat}
(P_{\sigma})_{ij} = 
\left\{
\begin{array}{ll}
     1 & \textup{if } j = \sigma(i);\\
     0 & \textup{otherwise}.
\end{array}
\right.
\end{equation}
The idea behind the construction of $\sigma$ is the following:
we show that, since $M(u)$ has an idempotent structure
and $\A$ is polynomially-ambiguous,
the non-zero entries of $M(u)$ define a partial order on the set of indices $\{1,2,\ldots,m\}$.
We then show that any permutation $\sigma$ which sorts these indices from largest to smallest
fits the desired requirements.

Formally, let us consider the binary relation
$\leq_{u}$ over $\{1,2,\ldots,m\}^2$
defined as
\begin{equation}\label{equation:order}
i \leq_{u} j
\textup{ if }
i = j
\textup{ or }
(M(u))_{ij} > 0.  
\end{equation}
We show that this relation is a partial order:
\begin{itemize}[nolistsep]
    \item 
    Reflexivity:
    this follows immediately from the definition;
    \item
    Transitivity:
    If $i \leq_{u} j$ and $j \leq_{u} k$ with $i \neq j \neq k$,
    then we get $(M(uu))_{ik} > (M(u))_{ij} \cdot (M(u))_{jk} > 0$
    as the entries of $M$ are non-negative.
    Since $M(u)$ has an idempotent structure, this implies that $(M(u))_{ik} > 0$,
    hence we have $i \leq_{u} k$;
    \item
    Antisymmetry: 
    suppose, towards building a contradiction,
    that there exist $i \neq j$ satisfying $i \leq_{u} j$ and $j \leq_{u} i$.
    This implies that $(M(uu))_{ii} > (M(u))_{ij} \cdot (M(u))_{ji} > 0$
    hence, as $M$ has an idempotent structure, $(M(u))_{ii} > 0$.
    Therefore, the automaton $\A$ has two distinct cycles on $i$ labelled by $vv$:
    one that loops twice on $i$ (witnessed by $(M(u))_{ii} > 0$),
    and one that goes to $j$ while reading the first copy of $v$ ($(M(u))_{ij} > 0$),
    and back to $i$ while reading the second one ($(M(u))_{ji} > 0$).
    This contradicts the fact that $\A$ is polynomially-ambiguous:
    Weber and Seidl's~\cite{WeberS91} showed that an automaton is not polynomially-ambiguous if and only if there exists a state $q$ and a word $w$ such that the there are at least two different cycles on $q$ labelled by $w$.
\end{itemize}
Let $\sigma$ be a permutation sorting $\{1,2,\ldots,n\}$
from largest to smallest according to $\leq_u$:
\begin{equation}\label{equation:permutation}
\textup{$i \leq_u j$ implies $\sigma(i) \geq \sigma(j)$ for every $1 \leq i,j \leq m$.}
\end{equation}
Finally, let $P_\sigma$ be the corresponding permutation matrix (as described by Equation~(\ref{equation:permutationMat})),
and let $M'= P_{\sigma}^{-1} \cdot M(u) \cdot P_{\sigma}$.
We conclude by showing that $M'$ is an upper triangular matrix:
For every $1 \leq i<j \leq m$ we have that $\sigma^{-1}(i) \not\leq_u \sigma^{-1}(j)$
(by the contrapositive of Equation~(\ref{equation:permutation})),
therefore $(M(u))_{\sigma^{-1}(i)\sigma^{-1}(j)} = 0$ by Equation~(\ref{equation:order}), and in turn:
\[
(M')_{ij}
 = (P_{\sigma}^{-1} \cdot M(u) \cdot P_{\sigma})_{ij}
= (M(u))_{\sigma^{-1}(i)\sigma^{-1}(j)}
= 0.\qedhere
\]
\end{proof}
\end{toappendix}

\begin{proof}[Proof of \cref{lemma:depump}]
We only prove Item~\ref{item:Pumpable},
Item~\ref{item:blindlyPumpable} can be proved nearly identically,
up to the quantifier alternation.

Suppose $\A$ is pumpable. We define $\ell$ as the constant in \cref{lemma:Ramsey}.
We define $\ramsey = \tower{r}(2\ell \cdot |\Sigma|)$ as the constant from \cref{lem:tower}.
Let $u,v \in \Sigma^*$ such that $|u| = \ramsey$.

We start by picking the decomposition $u = u_{\vdash} u_1 u_2 \ldots u_\ell u_{\dashv}$ from \cref{lem:tower}. By \cref{lemma:Ramsey} we can choose $1 \le i \le j \le \ell$ such that $M(u_i \ldots u_j)$ has an idempotent structure.
By Lemma~\ref{lemma:polyAmbTriangular} we know that $M(u_i \ldots u_j)$ is p-triangular. %
As $\A$ is pumpable, there is an entry $d \in \mathbb{N}$
of the diagonal of $M(u_i \ldots u_j)$ satisfying
\[
\A(u_1 \ldots u_{i-1} (u_i \ldots u_j)^n u_{j+1} \ldots u_{\ell}v) = d^{n - |\A|} \cdot \A(u_1 \ldots u_{i-1} (u_i \ldots u_j)^{|\A|} u_{j+1} \ldots u_{\ell}v)
\]
for all $n \geq |\A|$.
Let us rewrite the left-hand side as a product of matrices:
$I \cdot M(u_1\ldots u_{i-1}) \cdot M(u_i \ldots u_j)^n \cdot M(u_{j+1} \ldots u_{\ell}v) \cdot F$.
\cref{lem:tower} allows us to replace $M(u_i \ldots u_j)$ with an invertible p-triangular matrix $P$.
We get:
\begin{align*}
I \cdot M(u_1\ldots u_{i-1}) \cdot P^n \cdot M(u_{j+1} \ldots u_{\ell}v) \cdot F =
d^{n- |\A|} \cdot \A(u_1 \ldots u_{i-1} (u_i \ldots u_j)^{|\A|} u_{j+1} \ldots u_{\ell}v)
\end{align*}
for all $n \geq |\A|$.
As the right-hand side is composed of a single power $d^n$ multiplied by a constant polynomial
and $P$ is invertible and p-triangular,
we can rewrite it according to \cref{equation:invertible} of \cref{lemma:pumpMatrixS}:
\begin{align*}
I \cdot M(u_1\ldots u_{i-1}) \cdot P^n \cdot M(u_{j+1} \ldots u_{\ell}v) \cdot F = 
d^n \cdot I \cdot M(u_1\ldots u_{i-1}) \cdot M(u_{j+1}u_{\ell}v) \cdot F
\end{align*}
for all $n \in \mathbb{N}$.
We revert $P$ to $M(u_i \ldots u_j)$,
which yields the expression required by Lemma~\ref{lemma:depump} once we set $n = 1$:
\[
\A(u_1 \ldots u_{i-1} (u_i \ldots u_j)^n u_{j+1} \ldots u_{\ell}v) =
d^n \cdot \A(u_1 \ldots u_{i-1} u_{j+1} \ldots u_{\ell}v). \qedhere
\]
\end{proof}

\section{Deciding pumpability}\label{section:decision}

\newcommand{\pump}[1]{\mathcal{P}_{#1}}

This subsection is devoted to proving:
\propPumpIsDecidable*

\subsection{Deciding pumpability}

We start by showing pumpability is decidable, in fact, we will show that a seemingly weaker version of pumpability is decidable, however, we will observe that it is equivalent. 

Let us recall (from  \cref{def:pumpable}) that $\A$ is pumpable if for all $u,v,w \in \Sigma^*$ such that $M(v)$
    is p-triangular
    there is an entry $d$ of the diagonal of $M(v)$ 
    satisfying
\begin{equation}\label{equation:pumpMore}
    \A(uv^{m + n}w) = d^{n} \cdot \A(uv^{m}w)
    \textup{ for every }
    n \in \mathbb{N}. 
\end{equation}
However, weak pumpability requires only that \cref{equation:pumpMore} applies only for $n=1$:
\begin{definition}(Weak Pumpability)
    A weighted automaton $\A$ of size $m$
    is \emph{weakly pumpable}
    if for all $u,v,w \in \Sigma^*$ such that $M(v)$
    is p-triangular
    there is an entry $d$ of the diagonal of $M(v)$ 
    satisfying
    \begin{equation} \label{eq:weakpumpable}
    \A(uv^{m + 1}w) = d \cdot \A(uv^{m}w).    
    \end{equation}
\end{definition}

It is clear that if $\A$ is pumpable, then in particular
the weaker property is satisfied. We prove the converse implication: if  \cref{eq:weakpumpable} is satisfied by all triples $u,v,w$ such that $M(v)$ is p-triangular,
we show that each such triple also satisfies~\cref{equation:pumpMore}.
By applying Equation~\ref{eq:weakpumpable} to a family of the triples of the form $(uv^n,v,w)_{n \in \mathbb{N}}$,
we get inductively that for every $n \in \mathbb{N}$, there exist $m$ integers $j_1,j_2,\ldots,j_m$ summing up to $n$ such that
$\A(uv^{m+n}w) = \prod_{i=1}^m(M(v))_{ii}^{j_i} \cdot \A(uv^{m}w)$.
However, Lemma \ref{lemma:pumpMatrixS} gives us another expression for these weights:
$\A(uv^{m+n}w) = \sum_{i=1}^m(M(v))_{ii}^{n} \cdot p_i(n)$, where $p_i$ are polynomials.
The only way for these two expressions to match is that Equation~\ref{equation:pumpMore} holds,
which shows that $\A$ is pumpable. We prove this formally (proof in~\cref{appendix:weakIffpump}):

\begin{restatable}{lemma}{weakIffpump}
A weighted automaton $\A$ is weakly pumpable if and only if $\A$ is pumpable.
\end{restatable}
\begin{toappendix}
\label{appendix:weakIffpump}
\weakIffpump*
    
\begin{proof}
It is immediate that a pumpable automaton is weakly pumpable. We suppose that $\A$ is a weakly pumpable automaton and show $\A$ is pumpable.  Consider any triple $u,v,w$, with $M(v)$ p-triangular, then for every $n \in \N$ we have $\A(uv^{m+n+1}w) = d\cdot  \A(uv^{m+n}w)$ for some $d \in \{M(v)_{11},\dots,M(v)_{mm}\}$. To show that $\A$ is pumpable, we show that the same $d$ is chosen for every $n\in\N$.

Consider the unary weighted automaton $\U$ such that $\U(a^{n}) = \A(uv^{n}w)$ for every $n\in\N$. $\U$ is constructed with alphabet $\{a\}$, over the same states as $\A$ with initial vector $IM(u)$, final vector $M(w)F$ and $M_\U(a) = M(v)$. 

Note that, by weak pumpability of $\A$, the set of prime divisors of $\{ \A(uv^{n}w)\ |\ n \in \mathbb{N}\}$ is finite. Thus $\U$ is unambiguisable by the characterisation of \cref{theorem:Bell}.

Therefore by \cref{prop:UnambIsPump}, we have that $\U$ is pumpable. In particular, consider the triple $(\epsilon, a,\epsilon)$: we have that $M(a)$ is p-triangular, so there is some $d$ in the diagonal of $M(a)$ such that 
\[
\U(a^{m+n}) = d^n\U(a^m) \text{ for all } n\in\mathbb{N}
\]
and by definition of $\U$, this entails that 
\[
\A(uv^{m+n}w) = d^n \A(uv^mw) \text{ for all } n\in\mathbb{N}.
\]
Repeating for any $u,v,w$ we have $\A$ is pumpable.\qedhere

\end{proof}
\end{toappendix}

For every weighted automaton $\A$
we show how to construct a weighted automaton $\pump{\A}$
that has size exponential with respect to $\A$,
and maps every word to $0$ if and only if $\A$ is weakly pumpable.
Then, deciding weak pumpability of $\A$ amounts to deciding
zeroness of $\pump{\A}$, which can be done in polynomial space
(with respect to the size of $\A$).
The next lemma presents constructions tailored to the study of pumpable automata to be used as building blocks to construct $\pump{\A}$:
Given $\A$, we show how to construct automata that
recognise the triples $u,v,w$ such that $M(v)$ is p-triangular, and
that compute the value mapped by $\A$ to the words $(uv^nw)_{n \in \mathbb{N}}$ given only $u,v,w$.

\begin{restatable}{lemma}{lemmaConstructTBC}
    \label{lemma:iter}
    Let $\A= (Q,\Sigma, M, I, F)$ be a weighted automaton
    of size $m = |\A|$
    over $\Sigma$, and let $\$ \notin \Sigma$.
    \begin{enumerate}
        \item
        There exists an automaton $\mathcal{T}$ of size $2^{m^2} + 3$
        and norm 1
        satisfying
        \[
        \begin{array}{lll}
             \mathcal{T}(u\$v\$w) = 1 & 
             \textup{for all $u,v,w \in \Sigma^*$
             such that $M(v)$ is p-triangular};\\
             \mathcal{T}(u) = 0 &
             \textup{for all other
             $u \in (\Sigma \cup \{\$\})^*$}.
        \end{array}
        \]
        \item
          For all $n \in \mathbb{N}$ 
        there exist automata $\B_n$ of size $m^{2n}+2m$ and norm $||\A||^n$
        satisfying
        \[\begin{array}{lll}
             \B_n(u\$v\$w) = \A(uv^{n}w); &
             \textup{for all $u,v,w \in \Sigma^*$}.
             \end{array}
        \]
        \item
        For all $1\le i \le m$ there exists an automaton
          $\C_{i}$ of size $m + 2$
          and norm $||\A||$
        satisfying 
        \[\begin{array}{lll}
             \C_{i}(u\$v\$w) = (M(v))_{ii}  &
             \textup{for all $u,v,w \in \Sigma^*$}.
             \end{array}
        \]
    \end{enumerate}
    
\end{restatable}
\begin{remark}
    The values of $\B_n$ and $\C_i$ are unimportant and unspecified when the input does not take the form $u\$v\$w$ for $u,v,w\in\Sigma$.%
\end{remark}
\begin{toappendix}
\lemmaConstructTBC*
\begin{proof}[Proof of \cref{lemma:iter}]
\hfill
\begin{enumerate}
    \item Let $\T$ be the automaton with states $\{q_0,q_a\} \cup \{0,1\}^{Q\times Q}$, that is, other than $q_0,q_r$ and $q_a$, every state is a zero-one matrix indexed by $Q$.
Recall the notation $\overline{M}$, where $\overline{M}_{i,j} = 1$ if $M_{i,j} \ne 0$ and $0$ otherwise. 
Let the initial vector $I_\T$ be the zero vector, except that $I_\T(q_0) = 1$ and let the final vector be $F_\T$ be the zero vector, except that $F_\T(q_a) = 1$. 
Let $\T$ have the following transitions:
\begin{itemize}
    \item $q_0 \xrightarrow{a:1} q_0, \text{ and }   q_a \xrightarrow{a:1} q_a$ for each $a\in\Sigma$;
    \item $q_0\xrightarrow{\$:1} \id$, where $\id$ is the identity matrix over $Q$;
    \item For every $M \in \{0,1\}^{Q\times Q}$, $a\in\Sigma$, let $N = M \overline{M(a)} $, we have  $M \xrightarrow{a:1} \overline{N}$;
    \item If $M$ is p-triangular we have $M \xrightarrow{\$:1} q_a$;
\end{itemize}
Note that the automaton is deterministic such that $q_0\xrightarrow{ u\$v\$w:1}q_a$ if $M(v)$ is p-triangular, in which case the weight of the run is 1 as $F(q_a) = 1$, and otherwise no run reaches the final state.
Observe $|\T| =  2^{|Q|\times |Q|} + 2$ and $||\T|| = 1$.

\item $\B_n$ will operate in three components separated by reading $\$$.
The first and third components, reading $u$ and $w$ of $u\$v\$w$ respectively are direct copies of $\A$. The second component, reading $v$ but behaving like $v^n$, being the most interesting. 
For every state $q\in Q$, we let the copy $q^1$ of $q$ represents the state in the first component and a copy $q^3$ represents the state in the third component. The transitions in the first and third components are as in $\A$.

The second component simulates $n$ runs simultaneously using the state space $Q^{n-1} \times Q^{n}$. A state $(g_1,\dots,g_{n-1}, q_1,\dots,q_n)$, with $g_i,q_i\in Q$, denotes that the state of the $i$th run is $q_i$, and the guess that the $i$th run will terminate in state $g_i$ at the end of $v$, and thus the $i+1$st run will start in $g_i$. Only runs with the correct guess will proceed to the third component.

Between the first and middle component we have:
\[
q^1 \xrightarrow{\$:1} (g_1,\dots, g_{n-1},q, g_1,\dots,g_{n-1}) \text{ for all } g_1\dots g_{n-1} \in Q \text{ and } q\in Q.
\]
Within the middle component we have:
\[
(g_1,\dots,g_{n-1}, q_1,\dots,q_n) \xrightarrow{a: x_1\cdot \ldots \cdot x_n} (g_1,\dots, g_{n-1},q_1',\dots, q_n') \text{ for all } q_i \xrightarrow{a: x_i} q_i' \in \A.
\]
From the middle component to the third we have:
\[
(g_1,\dots,g_{n-1}, q_1,\dots,q_n) \xrightarrow{\$:1} q_n \text{ if } g_i = q_i \text{ for all } 1\le i\le n-1. 
\]
Let the initial and finial vectors, $I_{B_n}$ and $F_{B_n}$, be the zero vectors, except $I_{B_n}(q^1) = I(q)$ for every $q\in Q$ and $F_{B_n}(q^3) = F(q)$ for every $q\in Q$.

Consider a word $uv^n w$, where $u =u_1\dots u_{|u|}$, $v =v_1\dots v_{|v|}$ and $w =w_1\dots w_{|w|}$.
We observe, there is a bijection between runs in $\A$ on $uv^nw$ and runs in $\B$ on $u\$v\$w$.

The bijection is as follows:
Let the following be a run in $\A$ on $uv^nw$:
\begin{itemize}
    \item On $u$, the subrun $ q_{u,1}\to\ldots\to q_{u,|u|}\to q_{u,|u|+1}$ and having weights $I(q_{u,1}),x_{u,1},\ldots, x_{u,|u|}$,
    \item On $v^n$ the subrun $q_{v,1,1}\to\ldots q_{v,1,|v|+1} = q_{v,2,1}\to\ldots  \ldots\to q_{v,n,|v|+1}$, where $ q_{u,|u|}= q_{v,1,1}$, and having weights $x_{v,1,1},\ldots, x_{v,1,|v|}, x_{v,2,1},\ldots, x_{v,n,|v|}$,
    \item On $w$ the subrun $q_{w,1}\to\ldots\to q_{w,|w|+1}$, where $q_{v,n,|v|+1} = q_{w,1}$. and having weights $x_{w,1},\ldots x_{w,|w|+1}, F(q_{w,|w|+1})$.
\end{itemize}
This run is in bijection with the following run on $u\$v\$w$ in $\B_n$.
\begin{itemize}
\item On $u$ the subrun $q^1_{u,1}\to\ldots\to q^1_{u,|u|}\to q^1_{u,|u|+1}$ has weights $I(q^1_{u,1}),x_{u,1},\ldots, x_{u,|u|}$.
\item Which moves, with weight $1$ on $\$$ to $(q_{v,2,1},\ldots, q_{v,n,1},q_{v,1,1},\ldots, q_{v,n,1})$
\item On $v$ the subrun: 
$(q_{v,2,1},\ldots, q_{v,n,1},q_{v,1,1},\ldots, q_{v,n,1})\xrightarrow{}(q_{v,2,1},\ldots, q_{v,n,1},q_{v,1,2},\ldots, q_{v,n,2})\xrightarrow{}\ldots\xrightarrow{} (q_{v,2,1},\ldots, q_{v,n,1},q_{v,1,|v|+1},\ldots, q_{v,n,|v|+1})$
has weights
$(x_{v,1,1}\cdot\ldots\cdot x_{v,n,1}),\ldots, (x_{v,1,|v|}\cdot\ldots\cdot x_{v,n,|v|})$, and
\item Which moves, with weight $1$ on $\$$ to $q^3_{w,1} =  q_{v,n,|v|+1})$, and
\item On $w$ the subrun $q^3_{w,1}\to\ldots\to q^4_{w,|w|}\to q^1_{w,|w|+1}$, with weights $x_{w,1},\ldots, x_{w,|w|},F(q^3_{w,|w|+1})$.
\end{itemize}
Finally, we observe the weights have the same value: 
\begin{multline*}
    I(q_{u,1})\cdot x_{u,1}\cdot\ldots\cdot x_{u,|u|}\cdot x_{v,1,1}\cdot\ldots\cdot x_{v,1,|v|}\cdot x_{v,2,1}\cdot\ldots  \ldots\cdot x_{v,n,|v|}\cdot x_{w,1}\cdot\ldots x_{w,|w|+1} \cdot F(q_{w,|w|+1})
    = \\ 
    I(q^1_{u,1})\cdot x_{u,1}\cdot\ldots\cdot x_{u,|u|}\cdot1 \cdot  (x_{v,1,1}\cdot\ldots\cdot x_{v,n,1}) \cdot \ldots \cdot  (x_{v,1,|v|}\cdot\ldots\cdot x_{v,n,|v|}) \cdot 1\cdot x_{w,1}\cdot\ldots x_{w,|w|+1}\cdot F(q^3_{w,|w|+1}).
\end{multline*}

\item  $\C_i$ is a copy of $\A$ with two additional states $q_0,q_f$, in which $I_{\C_i}(q_0) =  F_{\C_i}(q_f) = 1$ (and all other entries of $I_{\C_i},F_{\C_i}$ are zero), and
\[
q_0 \xrightarrow{\$: 1} i,\quad i \xrightarrow{\$: 1} q_f,\text{ and}\]\[ q_0 \xrightarrow{a: 1} q_0 ,\quad  q_f \xrightarrow{a: 1} q_f \text{ for all }a\in\Sigma.\qedhere
\]
\end{enumerate}
\end{proof}
\end{toappendix}

We are now ready to show how to decide pumpability in polynomial space. 
Let $\A$ be a weighted automaton of size $m$,
let $\$$ be a fresh symbol that is not in the alphabet of $\A$,
and let $\mathcal{T}$, $(\B_n)_{n \in \mathbb{N}}$ and $(\C_{i})_{1 \leq i \leq m}$
be the weighted automata constructed in Lemma~\ref{lemma:iter}.
Using the result of \ref{lemma:combination} that the difference and product of functions recognised by weighted automata are themselves effectively expressible by weighted automata, let \[\pump{\A} = \mathcal{T} \cdot \prod_{i=1}^m (\B_{m+1} - \B_{m}\C_i).\]

\begin{lemma}
$\pump{\A}$ maps every word to $0$ if and only if $\A$ is weakly pumpable.
\end{lemma}
\begin{proof}
By definition, $\pump{\A}$ maps every word to $0$ if and only if for every triple $u,v,w \in \Sigma^*$
such that $\mathcal{T}(u\$v\$w) > 0$ (\ie $M(v)$ is p-triangular),
there exists $1 \leq i \leq m$ such that $\B_{m+1}(u\$v\$w) = \B_{m}(u\$v\$w)\C_i(u\$v\$w)$,
that is,
\[
    \textup{there exists $1 \leq i \leq m$ satisfying } \A(uv^{m+1}w) = (M(v))_{ii} \cdot \A(uv^{m}w).\qedhere
\]
\end{proof}
It remains to confirm that we can check whether $\pump{\A}$ maps every word to $0$ in space polynomial in $m$.
First, remark that the size of $\pump{\A}$ is exponential in $m$:
\[
|\pump{\A}| = 
|\mathcal{T}| \cdot \prod_{i=1}^m (|\B_{m+1}| + |\B_{m}||\C_i|) =
(2^{m^2}+3) \cdot (m^{4m+1})^m
\leq
m^{6m^2} \qquad \text{(assuming $m\ge 2$).}
\]
and the norm is also exponential in $|\A|$, this entails that the weight of any edge can be encoded in polynomial space:
\begin{equation}
    \label{eq:normcomp}
||\pump{\A}|| = 
||\mathcal{T}|| \cdot \prod_{i=1}^m \max\{||\B_{m+1}|| , ||\B_{m}|||\C_i||\} =
 1\cdot (||\A||^{m+1})^m = ||\A||^{m^2+m}.
\end{equation}

We show that zeroness of the exponential size automaton $\pump{\A}$ can be decided in PSPACE. It is known that deciding equivalence of $\Q$-weighted automata is in NC2~\cite{Tzeng96,kiefer13}. Recall, NC is the class of problems decidable using a circuit of polynomial size and polylogarithmic depth, with branching width at most two. Such problems can be solved sequentially in polylogarithmic space~\cite{Ruzzo81}.

In particular, this means the zeroness problem can be decided in polylogarithmic space in the size of the automaton. Thus applying the  polylogarithmic space zeroness algorithm to the automaton $\pump{A}$ requires only polylogarithmic space with respect to the exponential size automaton $\pump{\A}$, equivalently polynomial space with respect to $\A$. To conclude we show that any transition weight in $\pump{\A}$ can be computed in polynomial space, so that the equivalence procedure has access to any edge weight of $\pump{\A}$, and thus any bit of the representation of $\pump{\A}$, in PSPACE.

\begin{lemma}
    Given two states $q,q'$ of $\pump{\A}$ and $a\in \Sigma \cup \{\$\}$, the transition weight $M_{\pump{\A}}(a)_{q,q'}$ can be computed in polynomial space. 
\end{lemma}
\begin{proof}
    First observe the same is true for $\T,\B_{m+1},\B_{m}$ and $\C_i$ for all $1\le i\le m$. The automaton
    $\pump{\A}$ is built applying \cref{lemma:combination} polynomially many times. By induction, at each step the norm is at most exponential (in total bounded by~\cref{eq:normcomp}) and the weight can be computed recursively in polynomial space using the second part of \cref{lemma:combination}.
\end{proof}

\subsection{Deciding blind pumpability}
We show that we can also decide blind pumpability in PSPACE,
again by reducing to the zeroness problem for weighted automata.
Let us compare pumpability and blind pumpability:\begin{itemize}[nolistsep]
    \item A pumpable automaton requires, for any $(u,v,w)$ triple with $M(v)$ p-triangular,
    the existence of $d$ such that $\A(uv^{m + n}w) = d^n \cdot \A(uv^{m}w)$ for all $n \in \mathbb{N}$.
    \item A blindly pumpable automaton requires, for any $(u,v)$ pair with $M(v)$ p-triangular,
    the existence of $d$ such that $\A(uv^{m + n}w) = d^n \cdot \A(uv^{m}w)$ for all $w\in \Sigma^*$ and $n \in \mathbb{N}$.
\end{itemize}
That is, in order to be blindly pumpable an automaton needs to be pumpable,
and on top of that for any $(u,v)$ pair the choice of $d$ has to be the same for all $w$.
We show how to decide this property, assuming it is already known that $\A$ is pumpable using the previous section.
We start with the assumption that it is already known that $\A$ is pumpable, by first applying the algorithm of the previous section. We will then encode the requirement that the choice of $d$ be the same for any two suffixes into a zeroness problem (as in the previous section). In order to do this we we generalise the automata $T$ and $B_i$ of \cref{lemma:iter} to read two suffixes:

\begin{lemma}
    There exists an automaton $\T'$ such that $\T'(u\$v\$w\$w') = 1 $ if $M(v)$ is p-triangular,
    and $\T'$ is zero on all other inputs.
    There exist automata $\B_{1,n},\B_{2,n}$ such that $\B_{1,n}(u\$v\$w\$w') = \B_n(u\$v\$w)$ and $\B_{2,n}(u\$v\$w\$w')= \B_n(u\$v\$w')$.
\end{lemma}
\begin{proof}
    These automata are easily obtained by modifying the constructions in~\cref{lemma:iter}.
\end{proof}
\begin{lemma}
Let $m$ be the size of $\A$ and let $\mathcal{Q}_\A = \T'\cdot (\B_{1,m+1}\cdot\B_{2,m} - \B_{1,m}\cdot\B_{2,m+1})$. 
Then $\A$ is blindly pumpable if and only if $\A$ is pumpable and $\mathcal{Q}_\A$ maps every word to zero.
\end{lemma}
\begin{proof}
We verify for every $u,v$ and every $w,w'$ that the same $d$ is used to assert that 
$ \A(uv^{m + n}w) = d^n \cdot \A(uv^{m}w)$ and $ \A(uv^{m + n}w') = d^n \cdot \A(uv^{m}w')$. Since $\A$ is pumpable, the same $d$ is used for every $n$, so it is sufficient to verify the property only at $n =1$.
In particular, whenever $\A(uvw)\ne 0$, we require that \begin{equation}
\label{eq:blindlypumpableeq}    
\frac{ \A(uv^{m + 1}w)}{\A(uv^{m}w)} = \frac{ \A(uv^{m + 1}w')}{\A(uv^{m}w')} \text{ for all $w\in \Sigma^*$}.
\end{equation}
Note that, if $\A(uv^{m}w)=0$, then by pumpability $\A(uv^{m+n}w)=0$.
Then $\A(uv^{m+n}w)= d^n \A(uv^{m}w)=0$ for any $d$,
thus in particular the same choice of $d$ can be made as $\A(uv^{m+n}w')= d^n \A(uv^{m}w')=0$.

Observe that $\mathcal{Q}_\A$ is zero when either:
\begin{itemize}[nolistsep]
    \item $\T'$ is zero, that is $u,v,w,w'$ do not need to satisfy \cref{eq:blindlypumpableeq}, or,
    \item $\A(uv^{m+k}w)=0$, or $\A(uv^{m+k}w')=0$ for $k\in\{0,1\}$, or
    \item for every $u,v,w,w'$ \cref{eq:blindlypumpableeq} holds.\qedhere
\end{itemize}
\end{proof}
Like $\mathcal{P}_\A$, the size and norm of $\mathcal{Q}_\A$ is exponential, and we can test zeroness of $\mathcal{Q}_\A$ in PSPACE.
\section{Conclusion}
As mentioned in the introduction our PSPACE upper bounds are not constructive. If one would like to construct an equivalent deterministic automaton its size would be bounded by a tower of exponents (see \cref{remark:deterministic}, and recall that the constant $\ramsey$ is obtained from \cref{lem:tower}). We cannot extract the unambiguous automaton from our techniques\footnote{Note that it can be constructed, just not necessarily from our techniques: since from the decision procedure one can be sure of its existence, we can enumerate unambiguous weighted automaton and test each for equivalence.}. Recall that for unambiguisation in the proof of \cref{prop:pumpIsUnamb} we rely on Bell and Smertnig's result: \cref{theorem:Bell}. One could imagine a direct construction as we provide for the deterministic automaton. The issue is that for unambiguous automata one would need to keep track of all nonzero runs. Given an unambiguous automaton it is known that the number of such runs is bounded~\cite{WeberS91}, but the size of the unambiguous automaton could be arbitrary big.
A future work question is whether one could solve the determinisation and unambiguisation problems constructively in elementary time and space.

We do not provide any lower bound and we are unaware of such a result, even for the general class of weighted automata over rationals. We write a simple observation why obtaining lower bounds seems to be difficult:
Suppose one wants to encode a problem, \eg satisfiability of a SAT formula.
One would like to define a weighted automaton $\A$
that behaves like a deterministic (unambiguous) weighted automaton $\D$
except if the formula is satisfied,
which unlocks some nondeterministic (ambiguous) behaviour.
The issue is that all natural encodings can be verified with an equivalence query
to the deterministic (unambiguous) automaton $\D$, which over the rationals is in NC$^2$~\cite{Tzeng96}.


\begin{thebibliography}{10}

\bibitem{BarloyFLM22}
Corentin Barloy, Nathana{\"{e}}l Fijalkow, Nathan Lhote, and Filip Mazowiecki.
\newblock A robust class of linear recurrence sequences.
\newblock {\em Inf. Comput.}, 289(Part):104964, 2022.
\newblock \href {https://doi.org/10.1016/j.ic.2022.104964}
  {\path{doi:10.1016/j.ic.2022.104964}}.

\bibitem{bell2021noncommutative}
Jason Bell and Daniel Smertnig.
\newblock Noncommutative rational p{\'o}lya series.
\newblock {\em Selecta Mathematica}, 27(3):1--34, 2021.

\bibitem{abs-2209-02260}
Jason~P. Bell and Daniel Smertnig.
\newblock Computing the linear hull: Deciding deterministic? and unambiguous?
  for weighted automata over fields.
\newblock In {\em {LICS}}, pages 1--13, 2023.
\newblock \href {https://doi.org/10.1109/LICS56636.2023.10175691}
  {\path{doi:10.1109/LICS56636.2023.10175691}}.

\bibitem{Bell22}
Paul~C. Bell.
\newblock Polynomially ambiguous probabilistic automata on restricted
  languages.
\newblock {\em J. Comput. Syst. Sci.}, 127:53--65, 2022.
\newblock \href {https://doi.org/10.1016/j.jcss.2022.02.002}
  {\path{doi:10.1016/j.jcss.2022.02.002}}.

\bibitem{abs-2307-13505}
Yahia~Idriss Benalioua, Nathan Lhote, and Pierre{-}Alain Reynier.
\newblock Register minimization of cost register automata over a field.
\newblock {\em CoRR}, abs/2307.13505, 2023.
\newblock \href {https://arxiv.org/abs/2307.13505} {\path{arXiv:2307.13505}},
  \href {https://doi.org/10.48550/arXiv.2307.13505}
  {\path{doi:10.48550/arXiv.2307.13505}}.

\bibitem{BumpusHKST20}
Georgina Bumpus, Christoph Haase, Stefan Kiefer, Paul{-}Ioan Stoienescu, and
  Jonathan Tanner.
\newblock On the size of finite rational matrix semigroups.
\newblock In Artur Czumaj, Anuj Dawar, and Emanuela Merelli, editors, {\em 47th
  International Colloquium on Automata, Languages, and Programming, {ICALP}
  2020, July 8-11, 2020, Saarbr{\"{u}}cken, Germany (Virtual Conference)},
  volume 168 of {\em LIPIcs}, pages 115:1--115:13. Schloss Dagstuhl -
  Leibniz-Zentrum f{\"{u}}r Informatik, 2020.
\newblock \href {https://doi.org/10.4230/LIPIcs.ICALP.2020.115}
  {\path{doi:10.4230/LIPIcs.ICALP.2020.115}}.

\bibitem{ChistikovKMP22}
Dmitry Chistikov, Stefan Kiefer, Andrzej~S. Murawski, and David Purser.
\newblock The big-o problem.
\newblock {\em Log. Methods Comput. Sci.}, 18(1), 2022.
\newblock \href {https://doi.org/10.46298/lmcs-18(1:40)2022}
  {\path{doi:10.46298/lmcs-18(1:40)2022}}.

\bibitem{Colcombet15}
Thomas Colcombet.
\newblock Unambiguity in automata theory.
\newblock In Jeffrey~O. Shallit and Alexander Okhotin, editors, {\em
  Descriptional Complexity of Formal Systems - 17th International Workshop,
  {DCFS} 2015, Waterloo, ON, Canada, June 25-27, 2015. Proceedings}, volume
  9118 of {\em Lecture Notes in Computer Science}, pages 3--18. Springer, 2015.
\newblock \href {https://doi.org/10.1007/978-3-319-19225-3\_1}
  {\path{doi:10.1007/978-3-319-19225-3\_1}}.

\bibitem{CzerwinskiH22}
Wojciech Czerwinski and Piotr Hofman.
\newblock Language inclusion for boundedly-ambiguous vector addition systems is
  decidable.
\newblock In Bartek Klin, Slawomir Lasota, and Anca Muscholl, editors, {\em
  33rd International Conference on Concurrency Theory, {CONCUR} 2022, September
  12-16, 2022, Warsaw, Poland}, volume 243 of {\em LIPIcs}, pages 16:1--16:22.
  Schloss Dagstuhl - Leibniz-Zentrum f{\"{u}}r Informatik, 2022.
\newblock \href {https://doi.org/10.4230/LIPIcs.CONCUR.2022.16}
  {\path{doi:10.4230/LIPIcs.CONCUR.2022.16}}.

\bibitem{CzerwinskiLMPW22}
Wojciech Czerwinski, Engel Lefaucheux, Filip Mazowiecki, David Purser, and
  Markus~A. Whiteland.
\newblock The boundedness and zero isolation problems for weighted automata
  over nonnegative rationals.
\newblock In Christel Baier and Dana Fisman, editors, {\em {LICS} '22: 37th
  Annual {ACM/IEEE} Symposium on Logic in Computer Science, Haifa, Israel,
  August 2 - 5, 2022}, pages 15:1--15:13. {ACM}, 2022.
\newblock \href {https://doi.org/10.1145/3531130.3533336}
  {\path{doi:10.1145/3531130.3533336}}.

\bibitem{Daviaud20}
Laure Daviaud.
\newblock Register complexity and determinisation of max-plus automata.
\newblock {\em {ACM} {SIGLOG} News}, 7(2):4--14, 2020.
\newblock \href {https://doi.org/10.1145/3397619.3397621}
  {\path{doi:10.1145/3397619.3397621}}.

\bibitem{DaviaudJLMP021}
Laure Daviaud, Marcin Jurdzinski, Ranko Lazic, Filip Mazowiecki, Guillermo~A.
  P{\'{e}}rez, and James Worrell.
\newblock When are emptiness and containment decidable for probabilistic
  automata?
\newblock {\em J. Comput. Syst. Sci.}, 119:78--96, 2021.
\newblock \href {https://doi.org/10.1016/j.jcss.2021.01.006}
  {\path{doi:10.1016/j.jcss.2021.01.006}}.

\bibitem{droste2009handbook}
Manfred Droste, Werner Kuich, and Heiko Vogler.
\newblock {\em Handbook of weighted automata}.
\newblock Springer Science \& Business Media, 2009.

\bibitem{FijalkowR017}
Nathana{\"{e}}l Fijalkow, Cristian Riveros, and James Worrell.
\newblock Probabilistic automata of bounded ambiguity.
\newblock In Roland Meyer and Uwe Nestmann, editors, {\em 28th International
  Conference on Concurrency Theory, {CONCUR} 2017, September 5-8, 2017, Berlin,
  Germany}, volume~85 of {\em LIPIcs}, pages 19:1--19:14. Schloss Dagstuhl -
  Leibniz-Zentrum f{\"{u}}r Informatik, 2017.
\newblock \href {https://doi.org/10.4230/LIPIcs.CONCUR.2017.19}
  {\path{doi:10.4230/LIPIcs.CONCUR.2017.19}}.

\bibitem{halava2005skolem}
Vesa Halava, Tero Harju, Mika Hirvensalo, and Juhani Karhum{\"a}ki.
\newblock Skolem’s problem--on the border between decidability and
  undecidability.
\newblock Technical report, Citeseer, 2005.

\bibitem{HrushovskiOP018}
Ehud Hrushovski, Jo{\"{e}}l Ouaknine, Amaury Pouly, and James Worrell.
\newblock Polynomial invariants for affine programs.
\newblock In Anuj Dawar and Erich Gr{\"{a}}del, editors, {\em Proceedings of
  the 33rd Annual {ACM/IEEE} Symposium on Logic in Computer Science, {LICS}
  2018, Oxford, UK, July 09-12, 2018}, pages 530--539. {ACM}, 2018.
\newblock \href {https://doi.org/10.1145/3209108.3209142}
  {\path{doi:10.1145/3209108.3209142}}.

\bibitem{Jecker21}
Isma{\"{e}}l Jecker.
\newblock A ramsey theorem for finite monoids.
\newblock In Markus Bl{\"{a}}ser and Benjamin Monmege, editors, {\em 38th
  International Symposium on Theoretical Aspects of Computer Science, {STACS}
  2021, March 16-19, 2021, Saarbr{\"{u}}cken, Germany (Virtual Conference)},
  volume 187 of {\em LIPIcs}, pages 44:1--44:13. Schloss Dagstuhl -
  Leibniz-Zentrum f{\"{u}}r Informatik, 2021.
\newblock \href {https://doi.org/10.4230/LIPIcs.STACS.2021.44}
  {\path{doi:10.4230/LIPIcs.STACS.2021.44}}.

\bibitem{kiefer13}
Stefan Kiefer, Andrzej~S. Murawski, Jo{\"{e}}l Ouaknine, Bj{\"{o}}rn Wachter,
  and James Worrell.
\newblock On the complexity of equivalence and minimisation for q-weighted
  automata.
\newblock {\em Log. Methods Comput. Sci.}, 9(1), 2013.
\newblock \href {https://doi.org/10.2168/LMCS-9(1:8)2013}
  {\path{doi:10.2168/LMCS-9(1:8)2013}}.

\bibitem{KirstenM05}
Daniel Kirsten and Ina M{\"{a}}urer.
\newblock On the determinization of weighted automata.
\newblock {\em J. Autom. Lang. Comb.}, 10(2/3):287--312, 2005.
\newblock \href {https://doi.org/10.25596/jalc-2005-287}
  {\path{doi:10.25596/jalc-2005-287}}.

\bibitem{KlimannLMP04}
Ines Klimann, Sylvain Lombardy, Jean Mairesse, and Christophe Prieur.
\newblock Deciding unambiguity and sequentiality from a finitely ambiguous
  max-plus automaton.
\newblock {\em Theor. Comput. Sci.}, 327(3):349--373, 2004.
\newblock \href {https://doi.org/10.1016/j.tcs.2004.02.049}
  {\path{doi:10.1016/j.tcs.2004.02.049}}.

\bibitem{Kostolanyi22}
Peter Kostol{\'{a}}nyi.
\newblock Determinisability of unary weighted automata over the rational
  numbers.
\newblock {\em Theor. Comput. Sci.}, 898:110--131, 2022.
\newblock \href {https://doi.org/10.1016/j.tcs.2021.11.002}
  {\path{doi:10.1016/j.tcs.2021.11.002}}.

\bibitem{LombardyS06}
Sylvain Lombardy and Jacques Sakarovitch.
\newblock Sequential?
\newblock {\em Theor. Comput. Sci.}, 356(1-2):224--244, 2006.
\newblock \href {https://doi.org/10.1016/j.tcs.2006.01.028}
  {\path{doi:10.1016/j.tcs.2006.01.028}}.

\bibitem{Mohri97}
Mehryar Mohri.
\newblock Finite-state transducers in language and speech processing.
\newblock {\em Comput. Linguistics}, 23(2):269--311, 1997.

\bibitem{paz71}
Azaria Paz.
\newblock {\em Introduction to probabilistic automata}.
\newblock Academic Press, 1971.

\bibitem{Raskin18}
Mikhail~A. Raskin.
\newblock A superpolynomial lower bound for the size of non-deterministic
  complement of an unambiguous automaton.
\newblock In Ioannis Chatzigiannakis, Christos Kaklamanis, D{\'{a}}niel Marx,
  and Donald Sannella, editors, {\em 45th International Colloquium on Automata,
  Languages, and Programming, {ICALP} 2018, July 9-13, 2018, Prague, Czech
  Republic}, volume 107 of {\em LIPIcs}, pages 138:1--138:11. Schloss Dagstuhl
  - Leibniz-Zentrum f{\"{u}}r Informatik, 2018.
\newblock \href {https://doi.org/10.4230/LIPIcs.ICALP.2018.138}
  {\path{doi:10.4230/LIPIcs.ICALP.2018.138}}.

\bibitem{Reutenauer79}
Christophe Reutenauer.
\newblock On polya series in noncommuting variables.
\newblock In Lothar Budach, editor, {\em Fundamentals of Computation Theory,
  {FCT} 1979, Proceedings of the Conference on Algebraic, Arthmetic, and
  Categorial Methods in Computation Theory, Berlin/Wendisch-Rietz, Germany,
  September 17-21, 1979}, pages 391--396. Akademie-Verlag, Berlin, 1979.

\bibitem{Ruzzo81}
Walter~L. Ruzzo.
\newblock On uniform circuit complexity.
\newblock {\em J. Comput. Syst. Sci.}, 22(3):365--383, 1981.
\newblock \href {https://doi.org/10.1016/0022-0000(81)90038-6}
  {\path{doi:10.1016/0022-0000(81)90038-6}}.

\bibitem{Schutzenberger61b}
Marcel~Paul Sch{\"{u}}tzenberger.
\newblock On the definition of a family of automata.
\newblock {\em Inf. Control.}, 4(2-3):245--270, 1961.
\newblock \href {https://doi.org/10.1016/S0019-9958(61)80020-X}
  {\path{doi:10.1016/S0019-9958(61)80020-X}}.

\bibitem{Tzeng96}
Wen{-}Guey Tzeng.
\newblock On path equivalence of nondeterministic finite automata.
\newblock {\em Inf. Process. Lett.}, 58(1):43--46, 1996.
\newblock \href {https://doi.org/10.1016/0020-0190(96)00039-7}
  {\path{doi:10.1016/0020-0190(96)00039-7}}.

\bibitem{WeberS91}
Andreas Weber and Helmut Seidl.
\newblock On the degree of ambiguity of finite automata.
\newblock {\em Theor. Comput. Sci.}, 88(2):325--349, 1991.
\newblock \href {https://doi.org/10.1016/0304-3975(91)90381-B}
  {\path{doi:10.1016/0304-3975(91)90381-B}}.

\end{thebibliography}
\end{document}